\PassOptionsToPackage{colorlinks=true, citecolor=violet, linkcolor=blue, urlcolor=grey}{hyperref}
\documentclass[sigconf]{acmart}
\renewcommand\footnotetextcopyrightpermission[1]{} 

\usepackage{amsmath,amsbsy,amsfonts}
\usepackage{sansmath}
\usepackage{textcomp}
\usepackage{etoolbox}

\usepackage{graphicx}
\usepackage{subfigure}
\usepackage{color}
\usepackage{xcolor}
\usepackage{flushend}
\hypersetup{
  colorlinks,
  citecolor=violet,
  linkcolor=blue,
  urlcolor=gray
}

\usepackage{booktabs}
\usepackage{multirow}
\usepackage{threeparttable}
\usepackage{diagbox}
\usepackage{float}
\usepackage[linesnumbered, ruled]{algorithm2e}
\usepackage{setspace}
\usepackage{ulem}
\usepackage{caption}

\usepackage[subject={Top1}]{pdfcomment}

\usepackage{listings}
\usepackage{bbding}
\usepackage{pifont}
\colorlet{mygray}{black!30}
\colorlet{mygreen}{green!60!blue}
\colorlet{mymauve}{red!60!blue}

\usepackage{pythonhighlight}
\lstnewenvironment{mypython}[1][]{\lstset{style=mypython, breaklines=true, columns=fullflexible, breakatwhitespace=false, numbers=left,#1}}{}

\def\acmBooktitle#1{\gdef\@acmBooktitle{#1}}
\acmBooktitle{Proceedings of \acmConference@name
       \ifx\acmConference@name\acmConference@shortname\else
         \ (\acmConference@shortname)\fi}

\copyrightyear{2022}
\acmYear{2022}
\setcopyright{acmcopyright}
\acmConference[CCS '22] {Proceedings of the 2022 ACM SIGSAC Conference on Computer and Communications Security}{November 7--11, 2022}{Los Angeles, CA, USA.}
\acmBooktitle{Proceedings of the 2022 ACM SIGSAC Conference on Computer and Communications Security (CCS '22), November 7--11, 2022, Los Angeles, CA, USA}
\acmPrice{15.00}
\acmDOI{10.1145/3548606.3560565}
\acmISBN{978-1-4503-9450-5/22/11}

\newcommand{\revise}[1]{{#1}}

\newcommand{\rthread}[1]{{#1}}

\newcommand{\code}[1]{{\small\texttt{#1}}}
\newcommand{\smallcode}[1]{{\footnotesize\texttt{#1}}}
\newcommand{\para}[1]{\medskip\noindent\textbf{#1}}
\newcommand{\smallpara}[1]{\smallskip\noindent\textbf{#1}}

\newcommand{\etal}{\textit{et al. }}
\newcommand{\eg}{\textit{e.g., }}
\newcommand{\ie}{\textit{i.e., }}

\renewcommand{\emph}[1]{{\textit{#1}}}

\newcommand{\titlesysname}{{NFGen}\xspace}
\newcommand{\sysname}{{\small \sf \titlesysname}\xspace}

\newcommand{\f}{{$F(x)$}\xspace}

\newcommand{\addop}{{$\mathsf{+}$}\xspace}
\newcommand{\mulop}{{$\mathsf{\times}$}\xspace}
\newcommand{\gtop}{{$\mathsf{>}$}\xspace}
\newcommand{\divop}{{$\frac{1}{x}$}\xspace}
\newcommand{\expop}{{$e^{x}$}\xspace}
\newcommand{\logop}{{$\ln{x}$}\xspace}
\newcommand{\sqrtop}{{$\sqrt{x}$}\xspace}
\newcommand{\sigmoid}{{\textit{sigmoid}}\xspace}

\newcommand{\Pkm}{{$\mathcal{\hat{P}}$}\xspace}
\newcommand{\pkm}{{$\hat{p}_{k}^{m}$}\xspace}
\newcommand{\pk}{{$\hat{p}_{k}$}\xspace}
\newcommand{\kbar}{{$\bar{k}$}\xspace}

\newcommand{\zerohat}{{\hat{0}}\xspace}

\newcommand{\xhat}{{\hat{x}}\xspace}
\newcommand{\chat}{{\hat{c}}\xspace}
\newcommand{\phat}{{\hat{p}}\xspace}

\newcommand{\pkhat}{{$\hat{p}_{k}$}\xspace}

\newcommand{\repz}{{\emph{Rep2k}}\xspace}
\newcommand{\repf}{{\emph{RepF}}\xspace}
\newcommand{\repzps}{{\emph{Ps-Rep2k}}\xspace}
\newcommand{\repfps}{{\emph{Ps-RepF}}\xspace}
\newcommand{\shamir}{{\emph{Shamir}}\xspace}

\begin{document}



  


\title{\titlesysname: Automatic Non-linear Function Evaluation Code Generator for General-purpose MPC Platforms} 

\author{Xiaoyu Fan}
\email{fanxy20@mails.tsinghua.edu.cn}
\affiliation{IIIS, Tsinghua University}

\author{Kun Chen}
\email{chenkun@tsingj.com}
\affiliation{Tsingjiao Information Technology Co. Ltd.}

\author{Guosai Wang}
\email{guosai.wang@tsingj.com}
\affiliation{Tsingjiao Information Technology Co. Ltd.}

\author{Mingchun Zhuang}
\email{mczhuang@bupt.edu.cn}
\affiliation{Beijing University of Posts and Telecommunications}

\author{Yi Li}
\email{xiaolixiaoyi@tsingj.com}
\affiliation{Tsingjiao Information Technology Co. Ltd.}

\author{Wei Xu}
\email{weixu@tsinghua.edu.cn}
\affiliation{IIIS, Tsinghua University}

\renewcommand{\shortauthors}{Xiaoyu Fan et al.}

\begin{abstract}

Due to the absence of a library for non-linear function evaluation,  so-called \emph{general-purpose} secure multi-party computation (MPC) are not as ``general'' as MPC programmers expect.  Prior arts either naively reuse plaintext methods, resulting in suboptimal performance and even incorrect results, or handcraft \emph{ad hoc} approximations for specific functions or platforms.  We propose a general technique, \sysname\footnote{The source code is released in \url{https://github.com/Fannxy/NFGen}}, that utilizes pre-computed \emph{discrete piecewise polynomials} to accurately approximate generic functions using fixed-point numbers.  We implement it using a performance-prediction-based code generator to support different platforms.  Conducting extensive evaluations of 23 non-linear functions against six MPC protocols on two platforms, we demonstrate significant performance, accuracy, and generality improvements over existing methods. 


\end{abstract}



\maketitle

\section{Introduction}\label{sec:intro}

Privacy-preserving computation, especially \emph{secure multi-party computation (MPC)}, has attracted a lot of attention in both academia and industry.  They provide a promising trade-off between mining the data and privacy protection.  
People have proposed many \emph{general-purpose MPC platforms}~\cite{keller2020mp, li2019privpy, mohassel2017secureml, knott2021crypten, bogdanov2008sharemind, tan2021cryptgpu} that provide high-level abstractions and practical performance, allowing people to develop secure data processing applications without understanding the details of underlying MPC protocols.   

Most platforms use a version of \emph{secret sharing (SS)} protocols to build basic secure operations like \addop, \mulop, and comparison (e.g., \gtop), and then construct complex functions by composing them, just like writing plaintext expressions. The security of compound operations/functions is guaranteed by the \emph{universal composability}~\cite{canetti2001universally} of these protocols.  These platforms usually provide built-in support for common non-linear functions such as reciprocal (\divop, for real number divisions), exponential (\expop), logarithm (\logop), and square root (\sqrtop).   They implement these functions either using generic numerical methods (e.g. the \emph{Newton method}) or adopting protocol-specific algorithms like in~\cite{rathee2021sirnn,damgaard2019new}.

It remains a big challenge, however, to support the large variety non-linear functions in scientific computing and machine learning, such as $\chi^2 test$ and \sigmoid.  The naive approach is to compose them with built-in functions.  E.g., we can compute $tanh(x) = \frac{e^{x} - e^{-x}}{e^{x} + e^{-x}}$ by composing two \expop, one \divop and two \addop.  We refer to this approach as \emph{direct evaluation}. Unfortunately, there are four pitfalls.  

\smallpara{Pitfall 1: Correctness and precision. }
Most practical MPC platforms use \emph{fixed-point (FXP)} numbers instead of the common \emph{floating-point (FLP)} ones for efficiency.  Although there are attempts to support FLP in MPC~\cite{aliasgari2013secure}, the low performance for \addop prevents people from adopting it.  FXP still dominates the practical MPC platforms~\cite{keller2020mp, li2019privpy, mohassel2017secureml, knott2021crypten, tan2021cryptgpu, bogdanov2008sharemind}.  Unfortunately, ignoring the differences between FXP and FLP leads to two severe issues:

First, FXP supports a \emph{much smaller range and resolution} than FLP, leading to more overflow/underflow as well as precision loss.  Even worse, FXP cannot represent $NaN$ and $Inf$ like FLP.  Also, the inputs/outputs on MPC are in ciphertext, so there is no way to detect overflows.  Using $tanh$ as an example, even if the function has a range of $[-1, 1]$, the intermediate results, $e^{x}$ and $e^{-x}$, can easily overflow when $|x|$ is large.  In plaintext, people use a \emph{scaling conversion} to enlarge the range of FXP, but in MPC, the encrypted $x$ makes scaling costly. 
In fact, the built-in $tanh$ function in MP-SPDZ~\cite{keller2020mp} gives wrong results if $|x| > 44$, even if we increase FXP width to 128 bits.  
Second, each non-linear function has a precision loss that \emph{accumulates} if we compose them with multiple steps.  Section~\ref{exp:accuracy} shows more examples of both issues.

\smallpara{Pitfall 2: Performance. }
Even with aggressive optimizations, non-linear function evaluation takes significant time using MPC.  Depending on the platform, they can be orders of magnitude slower than the plaintext version.  We measure the relative performance between non-linear functions vs. basic operations in 6 MPC protocols and observe dramatic differences (Table~\ref{Tab:settings} in Section~\ref{settings}).  Composing these functions sequentially makes things even slower.  

\smallpara{Pitfall 3: Generality. }
Although we can write many non-linear functions using the built-ins, some functions are hard to implement.  For example, functions $\gamma(x, z)$, $\Gamma{(x, z)}$ and $\Phi{(x)}$ are defined as integrals.  It is very tedious, if not impossible, to implement them in MPC using numerical methods and built-in operations.  

\smallpara{Pitfall 4: Portability. }
Even if we can afford the engineering effort to build a complete scientific computation library, there are too many performance trade-offs to make it portable to different MPC systems and applications.  
MPC systems use a variety of protocols (to support different security assumptions), number representations and sizes, as well as custom programming languages.  Also, they are deployed in different hardware/software/network environments.  We want the computation library to maintain efficiency in all cases with minimal porting efforts. 

In this paper, we offer a new scheme, \emph{non-linear function code generator (\sysname)}, to evaluate non-linear functions on general-\\purpose MPC platforms.  We approximate each non-linear function using an $m$-piece \emph{piecewise polynomial} with max order $k$ (we automatically determine $k$ and $m$).  
Our approach has three advantages.   First, it only uses secure \addop, \mulop and \gtop operations supported by all popular MPC platforms, and we can prove that the security properties are the same as the underlying platform with the same adversarial models.  Second, the evaluation is \emph{oblivious}, i.e., the operation sequence is not dependent on input data, allowing for predictable running time and avoiding timing-related side-channels.  Third, obtaining the approximation is independent of input data and hence can be precomputed on plaintext.  

\smallpara{Key challenges.  } 
Finding a good \emph{$(k,m)$-piecewise polynomial} approximation for MPC platforms is a challenging problem. 
First, fitting polynomials for FXP computation introduces many challenges: 
1) the polynomial needs to meet both the \emph{range} and the \emph{resolution} requirements of FXP, making sure all intermediate steps neither overflow nor underflow; 
2) an FXP is essentially an integer, making the polynomials discrete. Fitting one polynomial minimizing the approximation error is an NP-complete \emph{integer programming (IP)} problem. Thus we need to find an approximation.
Second, there is a trade-off between $m$ and $k$: whether we get more pieces (mainly leading to more \gtop's) or use a higher-order polynomial (leading to more \mulop's). The choice depends on the performance of the specific MPC system, as we discussed above. 

\smallpara{Method overview. } We compute the polynomial fitting as a pre-computation step in plaintext.  
In a nutshell, we first construct a polynomial with order $k$ in FLP using \emph{Chebyshev Interpolation}~\cite{trefethen2019approximation} (or \emph{Lagrange Interpolation}~\cite{mathworld2015lagrange} for corner case) and then discretized it into FXP.  We check the accuracy of the FXP polynomial using random data samples.  If it does not achieve the user-specified accuracy, we split the input domain into two and recurse on each smaller domain.  
We design a series of algorithms leveraging the FLP capability to help find a better FXP approximation, like using two FXPs to expand the \emph{range} and improving the precision through residual functions.  Section~\ref{sec:NFCG-pkm} provides the details of the algorithms. 

At runtime, we evaluate the $(k,m)$-piecewise polynomial in an \emph{oblivious} way, i.e. the execution only depends on $(k,m)$, but not input data.  We design an \emph{oblivious piecewise polynomial evaluation (OPPE)} algorithm (Section~\ref{sec:oppe}).   
To get a good $(k, m)$ trade-off on different MPC platforms,   
\sysname uses a profiler to collect performance metrics of a specific deployment of an MPC platform and learns a model to predict the performance with different $(k, m)$.  \sysname automatically makes the choice using the prediction and generates MPC-platform-specific OPPE code using built-in code templates.  \sysname provides templates for both PrivPy~\cite{li2019privpy} and MP-SPDZ~\cite{keller2020mp}. 
The template is the only platform-dependent part in \sysname, and it only takes a short template to port \sysname to a new MPC platform. 

\smallpara{Evaluation results.  } We evaluate \sysname against 6 secret sharing protocols using $15$ commonly used non-linear functions on both PrivPy~\cite{li2019privpy} and MP-SPDZ~\cite{keller2020mp}.  
We observe significant performance improvements over baselines (direct evaluation) in $93\%$ of all cases with an average speedup of $6.5\times$ and a maximum speedup of $86.1\times$.  
\sysname saves $39.3\%$ network communications on average.
We also show that we can avoid the overflow errors and achieve much better accuracy comparing with the baseline, and allow a larger input domain even with small FXP widths.  
Using logistic regression (LR) as an example, we demonstrate performance and accuracy improvements over direct evaluation and \emph{ad hoc} approximations, with $3.5\times$ speedup and $0.6\%-14\%$ accuracy improvements.  Additionally, we illustrate how \sysname helps users to easily implement otherwise hard-to-implement functions on MPC by using 8 complex functions defined as integrals and the \emph{$\chi^{2}$ test} on real data.

In summary, our contributions include:

1) We propose a series of algorithms to fit an effective \emph{piecewise polynomial} approximation on plaintext, fully considering the differences between FLP and FXP representations.  

2) We design and implement the code generator with automatic profiling to allow portability to different MPC systems with distinct performance characteristics. 

3) We conduct comprehensive evaluations against six MPC protocols on two platforms over 23 non-linear functions, showing significant improvements in performance, accuracy, portability to different MPC platforms, usability and savings in communications. 

\begin{figure*}[t]
	\centering
	\includegraphics[width=0.9\textwidth]{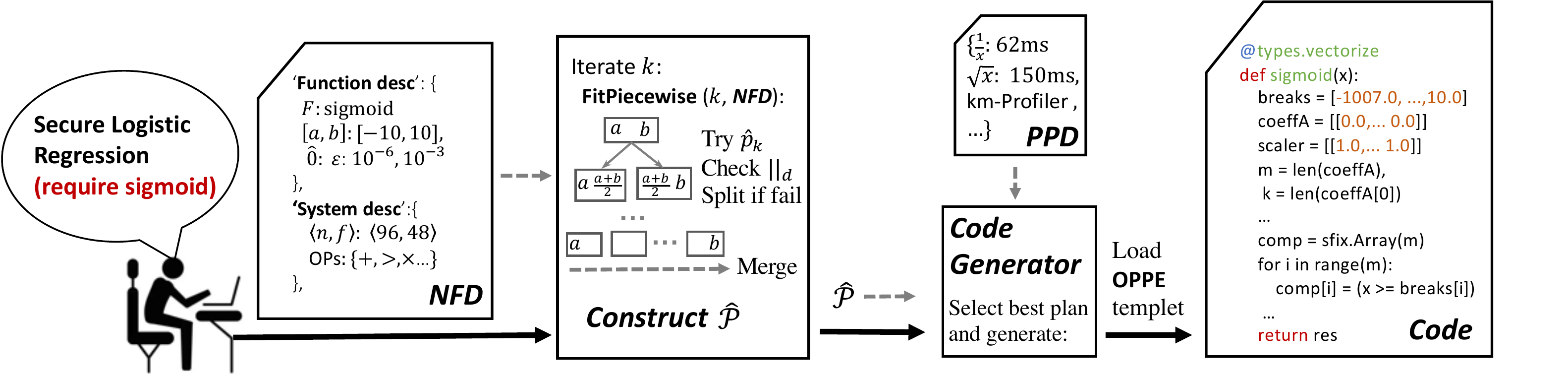}
	\vspace{-0.2in}
	\caption{End-to-end Workflow of \sysname}
	\label{Fig:workflow}
\end{figure*}

\section{Background and Related work}

In this section, we first introduce the background of general-purpose MPC platforms and FXP numbers for readers unfamiliar with this area and then review related works. 

\subsection{General-purpose MPC Platforms}\label{sec:back-knowledge}

General-purpose MPC platforms usually provide a high-level programming front-end with a compiler/interpreter to translate a high-level code into a series of cryptographic building blocks~\cite{hastings2019sok}.  
They provide common operators like \addop, \mulop, \gtop and \divop in the front-end, and implement these operations using MPC protocols.  These platforms guarantee privacy in the end-to-end algorithms based on the \emph{universal composability} of the underlying security protocols.  Example platforms include MP-SPDZ~\cite{keller2020mp}, PrivPy~\cite{li2019privpy}, CrypTen~\cite{knott2021crypten}, SecureML~\cite{mohassel2017secureml}, ABY~\cite{demmler2015aby}, and CryptGPU~\cite{tan2021cryptgpu} etc.  They hugely lower the barriers of developing privacy-preserving applications. 

Most of the general-purpose MPC platforms use \emph{secret sharing (SS)}~\cite{demmler2015aby, li2019privpy, knott2021crypten, mohassel2017secureml, tan2021cryptgpu} as the underlying protocol.  There are a range of security assumptions (e.g., semi-honest vs. malicious) in these protocols, resulting in significantly differing performance.  Some platforms allow users to choose the underlying protocols, like~\cite{keller2020mp}.  Our goal is to make \sysname protocol-agnostic.

Secret sharing protocols require communications among the participants on each operation, which results in significantly slower performance than plaintext (usually $10\times - 1000\times$ slower, depending on the protocol and network environment).  Also, different operators in the same protocol exhibit vast performance difference.  For example, in the most common \emph{additive} secret sharing protocol, \addop is almost as fast as plaintext as it is communication-free.  However, \mulop and \gtop require communication and are slower.  Non-linear operations like \divop and \sqrtop operate more slowly.  It is important to make \addop fast as it is the most common operation in many applications.

To make \addop fast, MPC systems need to avoid floating point (FLP) numbers because the ciphertext \emph{exponent} in FLP prevents us from adding up FLPs directly.  Thus, almost all practical MPC platforms use fixed-point (FXP) numbers~\cite{mohassel2017secureml, tan2021cryptgpu, knott2021crypten, mohassel2018aby3, li2019privpy, catrina2010secure}.  
FXP represents each real number as an $n$-bit integer, out of which $f$ least significant bits represent the fraction, and the most significant bit is the sign.  
We denote it as a $\langle n,f \rangle$-FXP number.  In the $\langle n,f \rangle$ format, both the \emph{range} and \emph{resolution} of the FXP are fixed to $2^{n-f-1}$ and $2^{-f}$, respectively.  
FXP offers a much smaller range and resolution than FLP, and thus programmers need to be more careful with overflows and precision losses.  MPC further complicates the problem as we cannot detect overflows on ciphertext.  Also, the common scaling method in plaintext FXP requires additional bit operations in ciphertext and thus becomes expensive in MPC.

\subsection{Non-linear Functions in MPC Platforms}

Lacking of a complete general numeric computation library, MPC application developers need to roll their application-specific solutions even for common functions, like logistic regression (LR)~\cite{han2020efficient}, decision trees~\cite{lindell2000privacy, damgaard2019new}, principal component analysis (PCA)~\cite{fan2021ppca} and neural networks (NNs)~\cite{rathee2021sirnn, tan2021cryptgpu, knott2021crypten}.  
Mohassel \etal~\cite{mohassel2017secureml, mohassel2018aby3} use a 3-piece linear function to replace the slow \sigmoid function in LR and NNs.
However, there is no guarantee that the approximation maintains accuracy in all cases.  In fact, we find cases with significant LR accuracy loss using this approximation. 

The most straightforward method to evaluate non-linear functions in MPC is directly adapting the plaintext code and replacing basic operations to secure ones (e.g., \addop, \mulop and \gtop).  For example, CrypTen~\cite{knott2021crypten} and CryptGPU~\cite{tan2021cryptgpu} adopt a series of plaintext algorithms like \emph{Newton-Raphson iterations}, \emph{limit approximations} to implement \divop, \expop, \sqrtop and \logop. 
Natively adopting plaintext code usually leads to poor performance, 
as we do not know when we have reached the desirable accuracy and need to iterate more. 

There are also protocol-specific approaches.  For example, Rathee \etal\cite{rathee2021sirnn} implement $10$ efficient cryptographic building blocks to support high-performance non-linear functions like \textit{reciprocal-of-sqrt}, \textit{sigmoid} and \textit{exponent}. However, these functions are optimized for their 2-party-computation platform only.  Similiarly, Damg{\aa}rd \etal\cite{damgaard2019new} propose a series of primitives to support efficient machine learning applications for SPD$\mathbb{Z}_{2^{k}}$ protocol~\cite{cramer2018spd}. 
More generally, Catrina \etal~\cite{catrina2010secure} design a collection of general building-blocks suitable for any fixed-point MPC platform and propose an optimized \divop primitive using \emph{Goldschmidt} algorithm on secret FXP. MP-SPDZ~\cite{keller2020mp} uses the same algorithm.

Another line of methods is to approximate non-linear functions with polynomials.  It is a well-studied problem in plaintext of finding an optimal polynomial approximating a given function \f, minimizing the maximum difference over a given input domain $[a, b]$.   
Such polynomial is referred to as \textit{minimax polynomial}~\cite{atkinson2005theoretical}.
\emph{Chebyshev interpolation}~\cite{trefethen2019approximation} is a well-known solution.  It offers a close approximation to the \emph{minimax polynomial} (so-called \emph{Chebyshev (near) minimax polynomial})~\cite{atkinson2005theoretical, trefethen2019approximation}.  Hesamifard \etal\cite{hesamifard2018privacy} adopt the Chebyshev polynomial to evaluate \sigmoid on homomorphic encryption.  However, the range of FXP limits the order $k$ of the polynomial and prevents it from reaching the desired approximation accuracy either. 
Another challenge is that the limited resolution of FXP cannot capture sometimes-tiny coefficients in these polynomials.  Previous work ignores this problem and leads to large errors~\cite{hesamifard2018privacy}.  Boura \etal~\cite{10.1007/978-3-662-58387-6_10} propose to use \emph{Fourier series} to approximate \sigmoid in MPC platform. They work around the range and resolution issues using secure quadruple-precision FLP but brings in expensive computation overhead.

Unlike previous efforts, we adopt \emph{piecewise polynomial} and fit a Chebyshev polynomial for each piece, resulting in improved accuracy.  
We take into account the differences between FXP and FLP and handle corner cases.  Also, prior works focus on one or a few functions on a single MPC platform, whereas our goal is to build a general, cross-platform solution for all \emph{Lipschitz-continuous} functions\footnote{Even if the functions are Lipschitz continuous, there is no theoretical guarantee that we will find a feasible approximation for \emph{all} functions on any accuracy requirements. However, empirically, \scriptsize{\sf NFGen} \footnotesize{works well on all the functions we tested.}}. 
\section{Overview}\label{sec:overview}

\smallpara{Notations and assumptions.  }
We first describe the notations and assumptions we use throughout the paper. 
We assume \f is a Lipschitz continuous function, and we can evaluate \f in plaintext using FLP. 
\sysname approximates \f by finding a set of feasible \emph{piecewise polynomial}s with different max order $k$ and the number of pieces $m$.  
We denote the set of all feasible polynomials as \Pkm $= \{\text{\pkm}\}$.  
Each of the \pkm contains $m$ pieces covering the entire input domain of $[a, b]$ as $[w_0, w_1, \dots,  w_j, \dots, w_m]$ ($w_0 = a$ and $w_m = b$).  In each piece $j, (i.e., [w_{j-1}, w_{j}])$, we have a polynomial $\hat{p}^{(j)}_{k}(\xhat) = \sum_{i=0}^{k}{\chat_{i}\xhat^{i}}$ to approximate it.  When the piece index $j$ is not important, we use a shorthand of \pk$(x)$ to denote it.  
Throughout the paper, symbols with a \emph{hat} (${\hat{\cdot}}$) denote $\langle n,f \rangle$-FXP representable variables and the definition of functions like \pkm means all the coefficients and terms are in $\langle n,f \rangle$-FXP representation (Section~\ref{sec:back-knowledge}).
We assume 64-bit double-precision FLP is equivalent to $\mathbb{R}$. It is a safe assumption in our situation, considering both the resolution and range of FLP are orders of magnitudes larger than FXP of our concern. 

\smallpara{\titlesysname input files.  }
\sysname generates non-linear function approximation code for different MPC platforms according to two input files.  The first is a user-provided \emph{non-linear function definition (NFD)}, containing the expression of the target function \f, its domain $[a, b]$, FXP format $\langle n, f\rangle$, target accuracy $\epsilon$ and $\zerohat$ (in Eq.~\ref{eq:rela-dis}), and a list of the operators supported by the target MPC platform.  The users can generate the second input file, \emph{performance profile definition (PPD)}, running a \sysname-provided profiler on the target MPC platform deployment.  Note that NFD is MPC-platform-specific but independent of the actual \emph{deployment} (i.e., the CPU and networking configurations), and PPD describes the deployment.  Separating them allows better portability across different deployments. 

\smallpara{Workflow.  }
Figure~\ref{Fig:workflow} illustrates \sysname workflow.  \sysname first reads in the NFD file, and on plaintext runs the algorithms in Section~\ref{sec:NFCG-pkm} to fit the set of \Pkm with different $k$ and $m$ settings.  Then using the PPD file, \sysname chooses one \pkm $\in$ \Pkm with the $k$ and $m$ that maximizes performance on the specific deployment (Section~\ref{sec:profiler}).  Finally, \sysname outputs the generated code that runs just like any user-defined function on the target MPC system, using a set of pre-defined, platform-dependent code templates.

\smallpara{Requirements for a feasible \pkm.  }
Given an MPC platform with $\langle n,f \rangle$-FXP, the target function \f on input domain $[a,b]$, a feasible \pkm needs to meet the following three conditions.  

1) The evaluation of \pkm should only consist of provided operators in the target MPC platform.  This requirement is usually true as all we need are basic operators of \addop, \mulop and \gtop.

2) All the intermediate results in \pkm, including polynomial coefficients $\chat_i$ and $\xhat^i$ terms, etc., should be representable in $\langle n,f \rangle$-FXP without overflow or underflow.  

3) For \emph{all} $\xhat$ in the range [$a$, $b$], the \pkm should approximate \f with high accuracy, so we need to bound the \emph{max error} of the approximation rather than the mean error.  We measure the approximation error using the \emph{soft relative distance (SRD)}, 
\begin{equation}\label{eq:rela-dis}
|x - y|_d = 
\begin{cases}
|x - y| / |x|, &  |x| > \hat{0} \\
|x - y|, & |x| \leq \hat{0}
\end{cases}, \hat{0} < \epsilon, 
\end{equation} 
and require
\begin{equation}~\label{cons: accuracy}
	\max_{\xhat \in [a, b]}|F(\xhat) - \text{\pkm}(\xhat)|_d \leq \epsilon. 
\end{equation}

The $\hat{0}$ is the \emph{soft zero}.   We use soft zeros because the relative error can be very large when $|x|\rightarrow 0$.  For example, for $x = 2^{-50}$ (in FLP), a good representable approximation in a $\langle96, 48\rangle$-FXP, $\xhat=2^{-48}$ gives a large relative error of $2^{2}-1 > \epsilon$.
To avoid ruling out these good-performance approximations, we switch to bound the \emph{absolute error} instead when $|x|\leq\hat{0}$. By default, we set $\epsilon=10^{-3}$ and $\hat{0}=10^{-6}$.
We further relax the accuracy definition to maximum \emph{sample} SRD rather than the true maximum SRD by computing the max error over a sample set of $\xhat \in [a,b]$ for practical performance.  Indeed, we prove that for any Lipschitz continuous function \f, the true maximum SRD can be bounded by the sampled version, for details, see Appendix~\ref{appendix:precision}.  Empirically, we find that a modest sample set (e.g. 1000 per piece) is frequently sufficient (Section~\ref{exp:accuracy}).  

\section{Non-linear approximation}
\label{sec:NFCG-pkm}

The core of \sysname is to fit a set of piecewise polynomials \Pkm that estimate the function \f in plaintext.  All \pkm $\in$ \Pkm can be evaluated obliviously in ciphertext. 
We first introduce the overall algorithm that recursively finds a good $m$ for a given $k$ if possible.  Then we focus on the algorithm that fits a $\hat{p}_{k}$ for a single piece, which is the most challenging part due to FXP limitations.  Finally, we introduce the oblivious evaluation algorithm.

\begin{algorithm}[t]
    \small
    \SetKwInOut{Input}{\textbf{Input}}
    \SetKwInOut{Config}{\textbf{Global config}}
    \SetKwInOut{GlobalState}{\textbf{Global state}}
    \SetKwInOut{Return}{\textbf{Return}}
    \SetKwFunction{Exit}{\textbf{Exit}}
    \caption{$\mathsf{FitPiecewise}$ Algorithm}
    \label{Algo:sampledpoly}
    \Config{FXP format $\langle n,f \rangle$, max sampling numbers $MS$ and max pieces $m_{max}$}
    \GlobalState{The fitted piecewise polynomial \pkm}
    \Input{Target function \f, input domain $[a, b]$ and order $k$ }
    \BlankLine
    Initialize piece counter $m_{c} \leftarrow 0$ \; 
	\If{$\hat{p}_{k} \leftarrow \mathsf{FitOnePiece}(F, [a, b], k)$ is NOT \code{Null}}{
        Add $\hat{p}_{k}$ to global state \pkm\;
    }\Else{
        $\mathsf{FitPiecewise}$($F$, $[a, \frac{a+b}{2}]$, $k$) \;
        $\mathsf{FitPiecewise}$($F$, $[\frac{a+b}{2}, b]$, $k$) \;
        $m_c += 1$\;
        \textbf{If} $m_c > m_{max}$: \textbf{Exit}\; 
    }
    \BlankLine
    $i \leftarrow 0$ \;
    \While{$i$ not reach the tail of \pkm}{
        $a, b \leftarrow w_{i}, w_{i+2}$ in \pkm \;
        \uIf{$\hat{p}_{k} \leftarrow \mathsf{FitOnePiece}(F, [a, b], k)$ is NOT \code{Null}}{
            Replace $\hat{p}_{k}^{(i)}$ and $\hat{p}_{k}^{(i+1)}$ with single $\hat{p}_{k}$ \;
        }\Else{
            $i += 1$ \;
        }
    }
    \textbf{Exit}
    \vspace{-0.05in}
\end{algorithm}

\subsection{Fitting Piecewise Polynomials}

The cost for evaluating \pkm is mostly for computing 1) the $k$-th order polynomial and 2) deciding which of the $m$ pieces that input $\xhat$ belongs to.  Each \pkm has $(m \times k)$ parameters.  We want to determine $m$ and $k$ automatically.  As different MPC platforms may have different \gtop and \mulop performance, we want to generate multiple plans with different $(k, m)$ choices and let the latter stages (Section~\ref{sec:profiler}) to decide which one to use.  
This step only takes 2-3 seconds on plaintext in most of our experiments.  

We iterate through a number of $k$ values.  For each $k$, we use Algorithm~\ref{Algo:sampledpoly} to find an $m$ piece polynomial as a candidate if possible.  

We start with the user-defined domain $[a, b]$.  We fit a best-effort $k$th-order polynomial \pk using \code{FitOnePiece} subroutine (Algorithm~\ref{Algo:fitonepoly}) minimizing the maximum absolute error (Line 2).  The fitting process is quite involved as we need to deal with the limited range and resolution of FXPs.  We leave the details of fitting \pk to Section~\ref{sec:single-poly}. 

\code{FitOnePiece} (Algorithm~\ref{Algo:fitonepoly}) returns a feasible \pk satisfying both representability and accuracy constrains in domain $[a, b]$ if it successfully finds it. If it failed, it returns \code{Null},  which means that the order $k$ is not enough to fit \f in domain $[a, b]$, and thus we split $[a, b]$ into $[a, (a+b)/2]$ and $[(a+b)/2, b]$ and recurse on each smaller ranges (Line 5-10).

As a final step, we try to merge adjacent pieces because splits may result in unnecessary pieces.  For each adjacent pair of pieces, we try \code{FitOnePiece} (Algorithm~\ref{Algo:fitonepoly}) again to fit a single polynomial in the combined range with the accuracy requirement satisfied (Line 11-19).  Finally, we get the set of all $m$ $k$th-order polynomials \pk's, constructing candidate \pkm.  

The algorithm eventually terminates, either when $m$ exceeds the limit $m_{max}$ (Line 9), or finds a \pkm that passes the accuracy test.  There is no theoretical guarantee that the algorithm will find a feasible solution or guarantee for the optimality.  However, it works well empirically on all functions in our tests (Section~\ref{exp:accuracy}).

\subsection{Fitting Polynomial for One Piece }\label{sec:single-poly}
We introduce the core part of Algorithm~\ref{Algo:sampledpoly}, the \code{FitOnePiece} function in Algorithm~\ref{Algo:fitonepoly}, fitting a single \pkhat in the domain of $[a, b]$.

\begin{algorithm}[tb]
	\small
	\SetKwInOut{Input}{\textbf{Input}}\SetKwInOut{Return}{\textbf{Return}}
	\SetKwFunction{Exit}{\textbf{Exit}}
	\SetKwFunction{constrainK}{$\mathsf{ConstrainK}$}
	\caption{$\mathsf{FitOnePiece}$ Algorithm}
	\label{Algo:fitonepoly}
	\Input{Target function $F(x)$, domain $[a, b]$ and order $k$}
	\Return{Feasible discrete polynoimial \pk or \code{Null}}
	\BlankLine
	$\bar{k} \leftarrow \mathsf{ConstrainK}([a, b], \langle n,f \rangle)$ \tcc{\textbf{1) Constrain k}}
	\BlankLine
	\tcc{\textbf{2) Fit best polynomial in FLP space}}
	Maximum representable points $\mathsf{N} \leftarrow \frac{b-a}{2^{-f}}$ \;
	\uIf{$\mathsf{N} > \bar{k}+1$}{
		$p_{\bar{k}} \leftarrow \textit{Cheby-Interpolation}(F, [a, b], \bar{k})$
	}\Else{
		$\bar{k} = N-1$ \;
		$p_{\bar{k}} \leftarrow \textit{Lagrange-Interpolation}(F, \mathsf{N}\, \text{feasible points and}\, \bar{k})$
	}
	\tcc{\textbf{3) \& 4) Convert to FXP space}}
	$\hat{p}_{\bar{k}} \leftarrow \mathsf{ScalePoly}(p_{\bar{k}}, [a, b])$ (Algo~\ref{Algo: scalepoly}) \;
	$\hat{p}_{\bar{k}} \leftarrow \mathsf{ResidualBoosting}(\hat{p}_{\bar{k}}, F, [a, b])$ (Algo~\ref{residual-boosting}) \;
	$\hat{p}_{k}$: Expand coefficients and scaling factors of $\hat{p}_{\bar{k}}$ to $k$, filling 0 \;	

	\tcc{\textbf{5) Check accuracy, return valid $\hat{p}_k$ or \code{Null}}}
	Sampled number $\mathsf{N_{s}} \leftarrow \min(MS, \mathsf{N})$ \;
	$\hat{\mathsf{X}} \leftarrow \mathsf{FLPsimFXP}(\mathsf{Linspace}([a, b], \mathsf{N_{s}}))$ \;
	\uIf{$\max_{\xhat \in \hat{\mathsf{X}}}|\hat{p}_{k}(\xhat) - F(\xhat)|_d < \epsilon$}{
		\textbf{Return: } $\hat{p}_{k}$ \;
	}
	\textbf{Return: } \code{Null} \;

\end{algorithm}

\smallpara{Problem Definition.  }
We want to find a $k$th-order polynomial $\hat{p}_k = \sum_{i=0}^{k}(\hat{c}_i \xhat^{i})$ that approximates \f over the domain of $\xhat \in [a, b]$, minimizing the max error.  
It is important that we limit the max error instead of the mean error to avoid occasional wrong results.  Formally, we define the following optimization problem. 
\begin{equation}~\label{eq:IQP}
\begin{aligned}
 \text{Minimize} & \max_{\xhat\in [a, b]}|F(\xhat) - \hat{p}_{k}(\xhat)| \\
 s.t., \, & \hat{p}_k = \sum_{i=0}^{k}(\hat{c}_i \hat{x}^{i}) \\ 
\end{aligned}
\end{equation}

It is easy to show that Eq.~\ref{eq:IQP} is an NP-Complete \emph{integer programming (IP)} problem, because 
\begin{equation}
\begin{aligned}
\hat{p}_k(\xhat) & = \sum_{i=0}^{k}\hat{c}_i \xhat^{i} =  \sum_{i=0}^{k}((\hat{c}_i \cdot 2^{f}) \cdot \frac{\xhat^{i}}{2^f}) 
 = \sum_{i=0}^{k}(\bar{c}_i \cdot \frac{\xhat^{i}}{2^f}), 
\end{aligned}
\end{equation}
where coefficients $\bar{c}_i$'s are $n$-bit integers.   We present an effective approximation by firstly solve the optimal polynomial $p_k$ in continuous space and then discretize it to $\hat{p}_{k}$ and optimize it in FXP.

\smallpara{Issues in FXP approximation.  }
~\label{sec: challenge}
FLP offers a much larger range compared with FXP. For $64$-bit double precision FLP, the representable range is from $-2^{1024}$ to $2^{1024}$ (IEEE754 standard~\cite{8766229}), while even for $\langle 128, 48\rangle$-bit FXP, the range is only $-2^{79}$ to $2^{79}$.  Thus, overflows are more common in FXP.  
For precision, double precision FLP can represent the smallest number of $2^{-1023}$, while FXP only has a fixed $f$-bit resolution. Any number smaller than $2^{-f}$ is rounded off to zero.
Unfortunately, prior MPC algorithms do not handle FXP correctly, leading to wrong results even if both the domain and range are representable. 

Specifically, We need to find a $\hat{p}_{k}= \sum_{i=0}^{k}\hat{c}_i\hat{x}^i$ in FXP to approximate $p_{k}=\sum_{i=0}^{k}c_ix^{i}$ in FLP and avoid the following three issues.  

\textbf{Issue 1)} $\hat{x}^k$ can overflow if $|\xhat|$ is too large, or underflow if $\xhat$ is close to zero, especially with a large $k$.  

\textbf{Issue 2)} When a coefficient $\hat{c}_i$ gets small, we need to use many of the $f$ bits to represent the leading $0$'s, losing significant bits, and even causing an underflow if $|\hat{c}_i| < 2^{-f}$.  However, $\xhat^i$ may be still large and we need an accurate $\hat{c}_i$ for $\hat{c}_i\xhat^i$ to approximate $c_ix^i$ in the continuous space.  
In fact, we observe that $c_i$ tends to be small when $i$ is close to $k$.  Intuitively, there is a relationship between the smoothness of target functions and their polynomial approximations.  The smoother the target function is, the faster its polynomial approximation converges (coefficient $|c_n| \rightarrow 0$ with $n \rightarrow \infty$). Theoretically,  \cite{trefethen2019approximation} shows that for Chebyshev polynomials, the absolute value of $k$th-order coefficient $|c_k|$ is inversely proportional to the exponent of its order $k$, presenting a quick descending rate. 

\textbf{Issue 3)} Converting all parameters into FXP involves many roundings to evaluate the polynomial (rounding the fractional parts beyond the $f$ bits to 0, in both computing $x$ to the $k$th power and adding-up all terms)\footnote{The conversion is done through \code{FLPsimFXP} (Algo~\ref{flpsimfxp}), which is analyzed in Appendix~\ref{proof:FLPsimFXP}}, hurting the precision.

\smallpara{Our solution.  }
Algorithm~\ref{Algo:fitonepoly} outlines our solution.  The algorithm firstly uses \emph{Chebyshev interpolation} or \emph{Lagrange interpolation} to find the optimal polynomial in the continuous space (represented by double-precision FLPs) and transfer the polynomial to discrete space (represented by FXPs) while avoiding the above issues.

\smallpara{Step 1: Constraining $k$ to avoid overflow (Issue 1). }
First, to avoid $\xhat^k$ overflow or underflow (issue 1),
we find the max feasible \kbar $\leq k$, guaranteeing that $\xhat^\text{\kbar} \text{does not overflow or underflow}\, \forall \xhat \in [a, b]$ (Line 1 in Algorithm~\ref{Algo:fitonepoly}).
In \code{ConstrainK} (Algorithm~\ref{Algo:constrainK}), overflow is easy to constrain, as we only need $(|\hat{x}|_{max})^{k_{O}} \leq 2^{n-f-1}$, or $k_{O} \leq \frac{n-f-1}{\log_2{|\hat{x}|_{max}}}$ if $|\hat{x}|_{max} > 1$ (Line 3).  Underflow is more involved, as if $0 \in [a, b]$, $|\xhat|$ can be arbitrarily close to zero.  In such case, we heuristically limit $k_{U}$ to $3$ (Line 4).  Otherwise, we need $(|\hat{x}|_{min})^{k_{U}} \geq 2^{-f}$, or $k_{U} \leq \frac{f}{(-\log_2{|\hat{x}|_{min}})}$ if $|\hat{x}|_{min} < 1$ (Line 5). The max feasible $\bar{k}$ is the smallest number among $\{k, k_{O} \,\text{and}\, k_{U}\}$.

\begin{algorithm}[tb]
	\small
	\caption{$\mathsf{ConstrainK}$}\label{Algo:constrainK}
	\SetKwProg{Fn}{Function}{:}{}
	\SetKwInOut{Return}{\textbf{Return}}
	\Fn{\constrainK{domain $[a, b]$, FXP format $\langle n,f \rangle$}}{
		\BlankLine
		$|\hat{x}|_{max} \leftarrow \max(|a|, |b|)$ and $|\hat{x}|_{min} \leftarrow \min(|a|, |b|)$ \;
		$k_{O} \leftarrow k \, \text{if} \, (|\hat{x}|_{max} < 1)\, \text{else}\, \frac{n-f-1}{\log_2{(|\hat{x}|_{max})}}$ \; 
		\textbf{If} $a \cdot b < 0$ \textbf{then} $k_{U} \leftarrow 3$\; 
		\textbf{Else} $k_{U} \leftarrow k \, \text{if} \, (|\hat{x}|_{min} > 1)\, \text{else}\, \frac{f}{-\log_2{(|\hat{x}|_{min})}}$ \;
		\Return{Maximum feasible $\bar{k} \leftarrow \min(k, k_{O}, k_{U})$}
	}
\end{algorithm}

\begin{algorithm}[tb]
	\small
	\SetKwInOut{Input}{\textbf{Input}}\SetKwInOut{Return}{\textbf{Return}}
	\caption{$\mathsf{FLPsimFXP}: x \rightarrow \hat{x}$}\label{flpsimfxp}
	\Input{$x \in \mathbb{R}$ and FXP format $\langle n,f \rangle$.}
	\Return{$\hat{x}$}
	\textbf{If} $(|x| > 2^{n-f-1})$: \textbf{Return: } $2^{n-f-1}$ \;
	\textbf{If} $(|x| < 2^{-f})$: \textbf{Return: } $0$ \;
	\textbf{Return:} $\mathsf{round}_2(x, f)$
\end{algorithm}

\begin{algorithm}[tb]
    \small
    \normalem
    \SetKwInOut{Input}{\textbf{Input}}\SetKwInOut{Return}{\textbf{Return}}
    \SetKwFunction{Initialize}{\textbf{Initialize}}
    \SetKwFunction{scalec}{$\mathsf{ScaleC}$}
    \caption{$\mathsf{ScalePoly}: p_{k} \rightarrow \hat{p}_{k}$}
    \label{Algo: scalepoly}
    \Input{$p_{k} = \sum_{i=0}^{i=k}(c_{i} \cdot x^{i})$ in continuous space, domain $[a, b]$.}
    \Return{$\hat{p}_{k} = \sum_{i=0}^{i=k}(\hat{c}_{i} \cdot x^{i} \cdot \hat{s}_{i})$ in discrete space.}
    \BlankLine
    Character $\xhat \leftarrow \max{(|a|, |b|)}$, most likely to overflow \;
    \For{$i \leftarrow 0$ \KwTo $k$}{
        $\hat{c}_{i}, \hat{s}_{i} \leftarrow \mathsf{ScaleC}(c_{i}, \langle n,f \rangle, i, \xhat)$ \;
    }
    \textbf{Return} $\hat{p}_{k}$.
    \BlankLine
    \SetKwProg{Fn}{Function}{:}{}
    \Fn{\scalec{\rm $c$, $\langle n,f \rangle$, order $k$ and $\xhat$}}{
    \BlankLine
    $\hat{s}_{CUF} \leftarrow 2^{-f}$ \;
    $\hat{s}_{COF} \leftarrow \mathsf{FLPsimFXP}(\frac{c \xhat^{k}}{2^{n-f-1}}, n, f)$ \;
    $\hat{s} \leftarrow \min{\{\max{(\hat{s}_{CUF}, \hat{s}_{COF})}, 1\}}$ \;
    $\hat{c} \leftarrow \mathsf{FLPsimFXP}(\frac{c}{\hat{s}}, n, f)$ \;
    \textbf{Return} $\hat{c}, \hat{s}$ \;
}
\end{algorithm}

\smallpara{Step 2: Fitting a polynomial in FLP.  }
We use \emph{Chebyshev interpolation}~\cite{trefethen2019approximation} that interpolates on $\bar{k}+1$ Chebyshev roots to construct the Chebyshev polynomial (Line 4), which is a close approximation to the real \emph{minimax polynomial}, i.e., minimizing the max approximation error. The method is widely adopted in practice (plaintext)~\cite{trefethen2019approximation, atkinson2005theoretical}. 
There is a corner case when $[a, b]$ is so small that $\frac{b-a}{2^{-f}} \leq \bar{k}+1$, i.e., the domain in such case does not contain enough points to fit the $\bar{k}$th-order polynomial. We construct the $(\frac{b-a}{2^{-f}}-1)$th-order polynomial trying to cover all the discrete points. Specifically, we use \emph{Lagrange interpolation}~\cite{mathworld2015lagrange} to solve the polynomial using all discrete points (Line 7).

Note that both Chebyshev and Lagrange methods fit the polynomial in continuous space using FLP.  Then we need to convert them back into the FXP space.  Issue 2-3 arise on this conversion. 

\smallpara{Step 3: Converting to FXP space with a scaling factor to enlarge the representation range (Issue 2).}
As we mentioned, when we round the fitted $c_i$'s to $\hat{c}_{i}$'s in FXP, 
$\hat{c}_i$ may be too small to represent precisely.  As $\hat{c}_i$'s are in plaintext, we can use the typical scaling factors to enlarge its representation range.
We translate each FLP $c_i$ into two FXP numbers, $(\chat_i,\hat{s}_i)$, letting $\chat_i \geq c_i$ to be large to preserve sufficient significant bits and $\hat{s}_i \leq 1$ is a scaling factor such that $c_i \approx \chat_i\hat{s}_i$.  
We want to let $\hat{c}_i$ contain more significant bits to maintain precision, but we need to avoid a too large $\chat_i$ that causes $\chat_i\xhat^k$ to overflow, especially when $\xhat^k$ is large.  More precisely, we require: 
1) $\hat{c}_i\xhat^k \leq 2^{n-f-1} $, i.e., it does not overflow; 2) $\hat{s}_i$ itself is a valid FXP, and 3) $ 0< \hat{s}_i \leq1$, i.e., it indeed scales up the coefficient, not making it even smaller.  
Algorithm~\ref{Algo: scalepoly} converts all FLP coefficients $c_i$ into $(\chat_i, \hat{s}_i)$ FXP pairs, satisfying all the three requirements. 

\smallpara{Step 4: Further reducing the rounding precision loss using residual boosting (Issue 2 - 3).}  
After step 3, we get an FXP approximation $\phat(\xhat) = \sum_{i=0}^{k}(\chat_i \xhat^{i} \hat{s}_i)$.
The rounding errors in $\hat{c}_i$ and $\hat{s}_i$ exacerbate the difference between the approximation and the real $F(x)$, and we use a \emph{residual function} $R(x) = F(x) - \phat_{k}(x)$ in FLP to capture the difference.  If we can estimate $R(x)$ with another discrete polynomial $\hat{r}_{k'}(x)$ with $k' < k$, we \emph{may} get a better precision if we use $\phat_{k}(x) + \hat{r}_{k'}(x)$ to approximate $F(x)$.  

\begin{algorithm}[tb]
    \normalem
    \small
    \SetKwInOut{Input}{\textbf{Input}}\SetKwInOut{Return}{\textbf{Return}}
    \SetKwFunction{boost}{$\mathsf{Boost}$}
    \caption{$\mathsf{Residual Boosting}$}
    \label{residual-boosting}
    \Input{$\hat{p}_k$ in discrete space, target \f and domain $[a, b]$.}
    \Return{$\hat{p}^{*}_{k}$ in discrete space.}
    \BlankLine
    $\hat{p}^{*}_{k} \leftarrow \hat{p}_k$, $R \leftarrow F - \hat{p}^{*}_{k}$ \;
    Sample number $\mathsf{N_{s}} \leftarrow \min{(\text{Max samples}\, MS, \text{All points} \, \frac{b-a}{2^{-f}})}$ \;
    $\hat{\mathsf{X}} \leftarrow \mathsf{FLPsimFXP}(\mathsf{Linspace}([a, b], \mathsf{N_{s}}))$ \;
    \For{$k' \leftarrow k-1$ \KwTo $0$}{
        $r_{k'} \leftarrow \textit{Cheby-Interpolation}(R, [a, b], k')$ \;
        $\hat{p}^{tmp}_{k} \leftarrow \mathsf{Boost}(\hat{p}^{*}_{k}, r_{k'}, [a, b])$\;
        \tcc{Boost when benefit exist. }
        \footnotesize
        \If{$(\max_{\hat{x} \in \hat{\mathsf{X}}}|F(\hat{x}) - \hat{p}^{tmp}_{k}(\hat{x})|_d < \max_{\xhat \in \hat{\mathsf{X}}}|F(\xhat) - \hat{p}^{*}_{k}(\xhat)|_d)$}{
            $\hat{p}^{*}_{k} = \hat{p}^{tmp}_{k}$, $R = F - \hat{p}^{*}_{k}$ \;
        }
    }
    \small
    \Return{$p^{*}_{k}$}
    \BlankLine
    \SetKwProg{Fn}{Function}{:}{}
    \Fn{\boost{\rm $p_k$ ,$r_{k'}, k\geq k'$ with domain $[a, b]$}}{
        \BlankLine
        $p_{k'}(x) = \sum_{i=0}^{k'}(\hat{c}^{(\hat{p}_{k})}_{i} \cdot \hat{s}^{(\hat{p}_{k})}_{i} + c^{(r_{k'})}_{i}) \cdot x^{i}$ \;
        $\hat{p}_{k'} \leftarrow \mathsf{ScalePoly}(p_{k'}, [a, b])$ \;
        \footnotesize
        $\hat{p}_{k}(x) = \sum_{i=0}^{k'}(\hat{c}^{(\hat{p}_{k'})}_{i} \cdot x^{i} \cdot \hat{s}^{(\hat{p}_{k'})}_{i}) + \sum_{i=k'+1}^{k}(\hat{c}^{(\hat{p}_{k})}_{i} \cdot x^{i} \cdot \hat{s}^{(\hat{p}_{k})}_{i})$ \;
        \small
        \Return{\pk}
    }
\end{algorithm}

We observe that we \emph{may} approximate $R(x)$ using a series of \emph{lower-order} FXP polynomials because the lower-order coefficients tend to be larger and preserve more significant bits.  Algorithm~\ref{residual-boosting} illustrates the residual boosting procedure.  In a nutshell, the algorithm iterates $k'$ through $(k-1)$ to $1$, trying to fit a $k'$th-order polynomial $\hat{r}_{k'}$ approximating $R(x)$ using the same Chebyshev interpolation as in Algorithm~\ref{Algo:fitonepoly} and $\mathsf{ScalePoly}$ Algorithm~\ref{Algo: scalepoly}.  
We add $\hat{r}_{k'}$ to $\hat{p}_k$ through function \code{Boost} in Algorithm~\ref{residual-boosting} if we are able to obtain smaller max error on sample set $\hat{\mathsf{X}}$.  
The residual boosting algorithm is best-effort and opportunistic, but empirically, it performs well (Section~\ref{sec:ablation-study}). 

\smallpara{Step 5: Checking if the polynomial is actually feasible and returning it if so.  }
As the last step, we get \emph{sample set} $\hat{\mathsf{X}}$ from $[a, b]$ in FXP and check the accuracy of $\phat_{k}(\xhat)$ obtained from the previous steps by computing $|\phat_{k}(\xhat)-F(\xhat)|_d, \forall \xhat \in \hat{\mathsf{X}}$ and find the max error.  If the check passes (max SRD less than $\epsilon$), we return the polynomial to Algorithm~\ref{Algo:sampledpoly}, otherwise the function returns a \code{Null}, causing Algorithm~\ref{Algo:sampledpoly} to recurse on smaller ranges.

\subsection{The Runtime Evaluation Algorithm OPPE}
\label{sec:oppe}
\begin{algorithm}[tb]
    \small
    \normalem
    \SetKwInOut{Config}{\textbf{Config}}
    \SetKwInOut{Input}{\textbf{Input}}\SetKwInOut{Return}{\textbf{Return}}
    \SetKwFunction{Initialize}{\textbf{Initialize}}
    \SetKwFunction{calculateKx}{$\mathsf{CalculateKx}$}
    \caption{\textit{OPPE} Algorithm ($\mathsf{OPPE}$)}
    \label{Algo:OPPE}
    \Config{Three parts plaintext parameters of \pkm: $\hat{W}$(without endpoint $\hat{w}_m$), $C=\{\hat{c}_{j,i}\}$ and $S=\{\hat{s}_{j,i}\}$.}
    \Input{Secret input $[\xhat]$.}
    \Return{The secret evaluation result of $[\text{\pkm}{(\xhat)}]$.}
    \BlankLine
    $[\mathsf{comp}] \leftarrow \mathsf{GT}([\xhat], \hat{W})$ \# compare $x$ with each break point. \;
    $[\mathsf{mask}] \leftarrow \mathsf{ADD([\mathsf{comp}], -\mathsf{leftshift}([\mathsf{comp}], 1))}$ \;
    \For{$i \leftarrow 0$ \KwTo $k-1$}{
        $[\mathsf{coeff}]_{i} \leftarrow \hat{\sum}_{j=0}^{j=m-1}\mathsf{MUL}([\mathsf{mask}]_{j}, \hat{c}_{j,i})$ \;
        $[\mathsf{scaler}]_{i} \leftarrow \hat{\sum}_{j=0}^{j=m-1}\mathsf{MUL}([\mathsf{mask}]_j, \hat{s}_{j,i})$ \;
    }
    $[\mathsf{xterm}] \leftarrow \mathsf{CalculateKx}([\xhat], k)$ \;
    \textbf{Return} $\hat{\sum}_{i=0}^{i=k-1}(\mathsf{MUL}(\mathsf{MUL}([\mathsf{coeff}]_{i}, [\mathsf{xterm}]_{i}), [\mathsf{scaler}]_{i})$\;
    \BlankLine
    \SetKwProg{Fn}{Function}{:}{}
    \Fn{\calculateKx{$[\xhat], k$}}{
        \tcc{Calculate $[1, [\xhat], [\xhat]^{2}, ..., [\xhat]^{k}]$}
        \BlankLine
        \footnotesize
        $\mathsf{shift} \leftarrow 1$, $[\mathsf{res}] \leftarrow [1, \mathsf{tail}([\xhat], k)]$ \tcc{repeat $[\xhat]$ $k$ times}
        \small
        \While{$\mathsf{shift} < k$}{
            $[\mathsf{res}]_{\mathsf{shift}:} = \mathsf{MUL}([\mathsf{res]_{\mathsf{shift}:}},\mathsf{[res}]_{:\mathsf{-shift}})$\;
            $\mathsf{shift}\, \times = 2$ \;
        }
        \textbf{Return} $[\mathsf{res}]$ \;
    }
\end{algorithm}

At runtime, we take the output of Algorithm~\ref{Algo:sampledpoly}, \pkm, as plaintext config, and take the secret-shared value $[\xhat]$ as ciphertext input ($[\xhat]$ indicate the secret shares of value $\xhat$ among each party), to compute the result $[\hat{p}_{k}^{m}[\xhat]]$ in the \emph{oblivious piece-wise polynomial evaluation (OPPE)} Algorithm~\ref{Algo:OPPE}.  
The piecewise polynomial \pkm is described by three parameters: $W = [a=\hat{w}_0, \hat{w}_1, \hat{w}_2, \dots \hat{w}_m=b]$ (without endpoint $\hat{w}_m$) are the boundaries for the $m$ pieces, the coefficients $\chat_{j,i}$ and scaling factors $\hat{s}_{j,i}$ for all $j = 0 \dots m-1$ and $i = 0 \dots k$.  
All plaintext and ciphertext inputs are FXP numbers. 
We use $\mathsf{ADD}$, $\mathsf{MUL}$ and $\mathsf{GT}$ in Algorithm~\ref{Algo:OPPE} to denote the subroutines evaluating \emph{secure} addition, multiplication and greater-than, respectively ($\hat{\sum}$ means summation through $\mathsf{ADD}$).

\smallpara{OPPE Design. }
OPPE Algorithm~\ref{Algo:OPPE} treats each subroutine as an \emph{arithmetic black box} and organize them \emph{obliviously}, \ie the execution path is independent of the inputs. For details of the obliviousness property, see Appendix~\ref{appendix:oblivious-oppe}.
OPPE first determines which piece $[\xhat]$ belongs to, using one vectorized $\mathsf{GT}$ and one $\mathsf{ADD}$ (Lines 1-2).  The comparison result is a ciphertext vector [\code{mask}], containing a single \emph{one}, and all other elements are \emph{zero}s. 
Line 3-6 select the coefficients and scaling factors obliviously, using the [\code{mask}].  
Line 7 computes all $[\xhat]^i\, \text{,} \forall i=0\dots k$ using $(\lfloor\log{k}\rfloor + 1)$ vectorized $\mathsf{MUL}$'s (each on two ciphertext vectors with size $< k$) using the subroutine \code{CalculateKx}.
Line 8 computes every term using two $\mathsf{MUL}$'s: $\mathsf{MUL}(\mathsf{MUL}(\hat{c}_{j,i}, \hat{x}^i), \hat{s}_{j,i})$, and adds up the products to compute the result.  Note that we must execute the two $\mathsf{MUL}$s in this specific order to take advantage of the scaling factor.

\smallpara{Complexity.  } Algorithm~\ref{Algo:OPPE} uses $(2km)$ plaintext-with-ciphertext $\mathsf{MUL}$'s, $m$ $\mathsf{GT}$'s and $O(k\log{k})$ $\mathsf{MUL}$'s.  We can also leverage the vector (a.k.a., SIMD) optimization in many MPC platforms. If so, we only need $(\lfloor\log{k}\rfloor + 1)$ rounds of ciphertext $\mathsf{MUL}$, $2$ rounds of plaintext-with-ciphertext $\mathsf{MUL}$, and 1 round of $\mathsf{GT}$.  Thus the running time is predictable on an MPC platform, independent of input $\xhat$. Appendix~\ref{sec:oppe-complexity} shows the complexity analysis of OPPE algorithm.

\smallpara{Independent operations and parallelism.}
We observe that the [\code{mask}] computation and coefficient selection step (Line 1-6) is independent of the \code{CalulateKx} routine (Line 7).  
Thus if an MPC platform supports concurrency, we can run both independently, further reducing the running time.  Also, when the input vector $\xhat$ is long, we automatically break it up into multiple pieces to utilize the underlying platform's threading support to evaluate each piece.

\section{Security Analysis}\label{sec:overview-security}
	
\smallpara{Security definitions.  } 
\sysname uses the same security definitions as the \emph{secure multi-party protocol} of underlying MPC platform, $\pi_{f}$ that aims to let $n$ parties evaluate function $f$ without a trusted third party.  The security is defined as the \emph{security properties} achieved in the presence of some adversary $\mathcal{A}$ who can control a set of at most $t$ corrupted parties according to some \emph{adversarial model}.  Different protocols have their own choices of both the adversarial model and security properties to achieve, usually for trade-off of performance.  

The common \emph{adversarial model} trade-offs include 4 dimensions\\~\cite{lindell2017framework}: 
1) \emph{corruption strategy}: adaptive vs. non-adaptive; 
2) \emph{corruption proportion}: dishonest vs. honest-majority;
3) \emph{behavior}: malicious vs. semi-honest;
4) \emph{power}: informational vs. computational-secure.

Under these assumptions, protocols usually achieve the following two essential 
security properties: 
\emph{privacy} and \emph{correctness}. 
Optionally, there are other security properties a protocol may consider~\cite{lu2019honeybadgermpc}. E.g.,  
\emph{fairness}, i.e., if one party receives the result, all parties receive it; and \emph{guaranteed delivery}: whether the joint parties can always receive the results. 

\smallpara{Security assumptions of \titlesysname. }
\sysname builds on top of general-purpose MPC platforms with each party carrying out the computation connected through \emph{secure channels}. 
\sysname assumes that three secure subroutines $\mathsf{ADD}$, $\mathsf{MUL}$ and $\mathsf{GT}$ evaluating secret addition, multiplication and greater-than are provided and all the inputs are secret-shared among computation parties before the evaluation of the generated protocol. These assumptions are easily to meet as various implementations of secret addition, multiplication and greater-than have been proposed and some of them are widely adopted (\eg \cite{beaver1991efficient,damgaard2006unconditionally}).
Specifically, \sysname introduces no different assumptions about the \emph{adversary model} for the underlying protocols that implement the three required subroutines. 

\smallpara{Security analysis of \titlesysname.  } 
We illustrate the security of \sysname's code generation approach by showing that it guarantees the \emph{same security properties} as the three subroutines in the presence of the same adversary.

1) Obviously, the pre-computation steps (Algorithm ~\ref{Algo:sampledpoly} - ~\ref{residual-boosting}) are independent of secret inputs and performed offline, thus cannot affect any security property;

2) To prove the privacy and security properties, we directly follow the \emph{real-ideal} paradigm introduced in~\cite{canetti2000security}.  It defines security by requiring that the distribution of the protocol evaluation in the \emph{real} world is indistinguishable from the \emph{ideal} world with a trusted third party. 
Under this paradigm, we show that
\sysname generates the \pkm evaluation protocol by \emph{composing} the three subroutines as so-called arithmetic black boxes in the standard \emph{modular composition} way~\cite{canetti2000security} without revealing any information nor introducing any interaction.  Thus, \sysname naturally inherits the same security property of the subroutines from the \emph{Canetti's composition theorem}~\cite{canetti2000security}. The security preserving property is the direct result of the composition theorem. We show the detailed analysis in Appendix~\ref{sec:formal-security}.

3) We show that \sysname provides the same optional security properties as the underlying protocols.  Using \emph{guaranteed delivery} as an example, if the provided subroutines offer this property, meaning that there is no \emph{abort} within these subroutines, the \pkm evaluation protocol (Algorithm~\ref{Algo:OPPE}) will not abort either, as there is no breakpoint in the routine.  On the other hand, if the subroutines are \emph{secure with abort}, Algorithm~\ref{Algo:OPPE} does not try to handle these abortions at all and lets the protocol abort.  

\section{Implementation }
\label{sec:methods}

In this section, we briefly introduce how we integrate the fitted \Pkm into the target MPC system.

\subsection{Profiler and Performance Prediction}\label{sec:profiler}

In order to select the best \pkm $\in$ \Pkm, we need to model the performance of specific deployed MPC systems. Such modeling is not straightforward as the performance not only depends on the MPC protocols but also on the implementation and deployment. We use a profiler that runs on the target system to automatically build the performance prediction model.  

The execution time of Algorithm~\ref{Algo:OPPE} depends on $(k,m)$ only, making the prediction possible.  
We use the profiler to measure the evaluation time $t$ of $2,000$ configurations of piecewise polynomial samples with different $(k,m)$ combinations ($k \in [3, 10]$ and $m \in [2, 50]$).
Then we fit a \textit{multivariate polynomial regression model}~\cite{sklearn2011Regression} on these samples.  
Note that the performance model is independent of $F(x)$, and thus we only need to profile once per system. 

Some MPC systems provide built-in functions that are highly optimized for their settings, such as \divop, \expop, or \logop.  For example, MP-SPDZ provides a very efficient \divop~\cite{catrina2010secure}.  If users list these functions in the NFD, the profiler also measures the performance of such functions. 
If \f is simple, it may be better off taking the direct evaluation approach.  Direct evaluation is only viable with all the three conditions: 
1) \f does \emph{not} contain \expop as an intermediate step, as it is highly likely to overflow; 
2) \f contains less than three steps with non-linear functions to avoid unpredictable error accumulations and
3) the predicted running time of direct evaluation is shorter than all \pkm $\in$ \Pkm.  
We rarely find suitable cases to use direct evaluation, but on functions like $isru$, which is $1/(1 + |x|)$, effectively just a single \divop, direct evaluation is $1.6\times$ faster than the best \pkm (Section~\ref{micro-benchmark}).

We can either run the profiler in pre-computation or run it \emph{just-in-time} right before running a large MPC task.  In this paper, we do all profiling and plan selections in the pre-computation.

\subsection{OPPE Code Generation  }

Different MPC platforms offer not only different high-level languages, but also different support for vector operations and multi-threading.  To best utilize these platform-specific optimizations while still remain portable, we use a template-based code generation approach to implement OPPE (Algorithm~\ref{Algo:OPPE}).  

\sysname provides MPC-platform-specific code templates implementing OPPE.  Each template is highly optimized for a specific platform.  For example, the code template uses multi-threading in PrivPy to compute all independent segments concurrently~\cite{li2019privpy}, and leverages the \textit{probability trunction} optimizations in MP-SPDZ~\cite{dalskov2020secure}.  Note that we only need to customize the OPPE template for each platform, and all other procedures in \sysname are reuseable across platforms.  Currently, we support both PrivPy and MP-SPDZ. 

Using the templates, it is straightforward to generate \pkm evaluation code, as we only need to insert \pkm parameters in to the code template as literals.  Appendix~\ref{sec:code-example-NFCG} shows an example of NFD, PPD and generated code.  
The generated code runs the same as normal functions in the target MPC system.  

We pass \pkm as literals in generated code instead of arguments to OPPE function, because compilers in platforms like MP-SPDZ significantly improves performance if all input lengths (in our case, $m$ and $k$) are statically known.
\section{Evaluation}

In our evaluation, we show that \sysname is able to 1) offer better performance and lower communication costs across algorithms, protocols and systems; 2) avoid the overflow/underflow errors in the traditional approaches, and provide better accuracy in complex functions; 3) calculate sophisticated non-linear functions otherwise requiring extensive calculus knowledge to implement; 4) support a large domain with reasonable accuracy even with a very limited number of bits; and 5) benefit real applications with both performance and accuracy improvements.  
We also evaluate the different design choices in \sysname, such as the effectiveness of profiling, scaling and residual boosting. 

\subsection{Experiment Setup}~\label{settings}

\smallpara{MPC platforms. }
We evaluate \sysname on two MPC platforms, MP-SPDZ~\cite{keller2020mp} that implements over 30 secret sharing protocols (and we choose 5 for evaluation), and PrivPy~\cite{li2019privpy} that only supports a single protocol.  MP-SPDZ first compiles the high-level code into bytecode to execute with underlying protocols, while PrivPy executes programs by interpreting Python code at runtime.  Both platforms support FXP number with different width. 

\smallpara{Secret sharing protocols. }
We adopt six different secret-sharing protocols, covering different security assumptions over adversarial behaviors (\emph{semi-honest} or \emph{malicious}); \emph{computation domains} (over a ring of $\mathbb{Z}_{2k}$ or finite field $\mathbb{F}_{p}$ by modulo a prime $p$) and \emph{sharing methods} (using replicated secret sharing or \emph{shamir} secret sharing). We briefly introduce each protocol in the following.

First we introduce the four \emph{semi-honest} protocols:
A. \emph{PrivPy-Rep2k} is an \emph{2-out-of-4} replicated secret sharing protocol, proposed by Li \etal~\cite{li2019privpy}. It splits each value $x$ into four shares over a ring of $\mathbb{Z}_{2k}$ and let each party $P_i$ ($i = {1, 2, 3, 4}$) holds two shares, satisfying that any two parties can reconstruct $x$ while each one sees two random integers. 
B. \emph{Rep2k} and C. \emph{RepF} two protocols split a value $x$ into three shares, satisfying that $x \equiv x_1 + x_2 + x_3 \pmod{M}$ and let each party $P_i$ ($i = {1, 2, 3}$) holds $(x_i, x_{i+1})$(indexes wrap around 3). $M = 2^{k}$ for \emph{Rep2k} protocol and $M = p$ for \emph{RepF} protocol where $p$ is a prime.
D. \emph{Shamir} shares a value $x \in \mathbb{F}$ though a random chosen 2-degree polynomial $f_{s}$, such that $f_s(0) = x$. Each party $P_i$, ($i = {1, 2, 3}$) holds a distinct point over polynomial $f_{s}$. They together can reconstruct $f_{s}$ and obtain $x = f_s{(0)}$ while any set less than three parties contain no information about $x$.

Then we introduce the two \emph{malicious} protocols:
E. \emph{Ps-Rep2k} and F. \emph{Ps-RepF}: \emph{Ps} refers to \emph{Post-Sacrifice} strategy proposed by Lindell \etal~\cite{lindell2017framework}. It compiles a semi-honest protocol into a malicious secure version by adding a verification step. The verification step let the honest parties detect cheating behavior with high probability. The initial work~\cite{lindell2017framework} only considers finite field (\emph{Ps-RepF}) and the follow-up work~\cite{cramer2018spd} extends it to ring (\emph{Ps-Rep2k}).

We choose these six protocols not only because they are common practical protocols with different assumptions, but also because they exhibit various performance characteristics on the basic operations, which is helpful to show the generality of \sysname.  
Table~\ref{Tab:settings} summarizes the selected MPC settings ($\mathcal{S}$) with the absolute performance of \mulop and the performance of the other basic operations relative to \mulop.

\begin{table}
	\small
	\centering
	\caption{\revise{MPC Settings with Varied Operation Performance}}\label{Tab:settings}
	\vspace{-0.2in}
	\begin{threeparttable}
		
		\begin{tabular}{ c | c | c | c| l  }
			\midrule[1.1pt]
			Sec model & No. & MPC sys $\mathcal{S}$ & \mulop(ms) & \mulop : \gtop : \divop : \sqrtop : \logop : \expop\\ \hline\hline
			\multirow{4}{*}{\shortstack{Semi\\-honest}} & A & \textit{PrivPy Rep2k} & 1 & $1:11:67:55:118:44$\\ \cline{2-5}
			&B & \textit{Rep2k} & 2 & $ 1:4:31:75:68:107$\\ \cline{2-5}
			&C & \textit{RepF} & 32 & $1:1:11:28:26:47$\\ \cline{2-5}
			&D & \textit{Shamir} & 81 & $1:1:8:16:15:29$ \\ \hline
			\multirow{2}{*}{Malicious} & E & \textit{Ps-Rep2k} & 851 & $1:1:16:35:26:97$\\ \cline{2-5}
			&F & \textit{Ps-RepF} & 84 & $1:2:24:56:44:137$\\ \hline
		\end{tabular}
		\footnotesize
		\begin{tablenotes}
			\item B-E use $\langle 96,48 \rangle$-FXP and A uses $\langle 128,48 \rangle$. Column 3 is the absolute time (in ms) to compute \mulop on 100-dimensional vector and Column 4 is the relative performance to \mulop.
		\end{tablenotes}
	\end{threeparttable}
\end{table}

\smallpara{Evaluation environment. }
We perform all the evaluations on a cluster of four servers with two 20-core 2 GHz Intel Xeon CPUs and 180 GB RAM each, connected through 10 Gbps Ethernet.  MP-SPDZ uses only three servers while PrivPy uses all four.

\subsection{Performance}
~\label{micro-benchmark}

\begin{table*}
\caption{Examples in Performance Evaluations (Full Results in Appendix~\ref{sec:full-microbenchmark}) }\label{Tab:mb-activation}
\vspace{-0.2in}
\scriptsize
\resizebox{0.95\textwidth}{!}{
\begin{threeparttable}
    \centering
    \begin{tabular}{c|c|c|c|c|c|c|c|c|c|c}
    \midrule[1.1pt]
    \multirow{2}{*}{$F(x)$} & \multirow{2}{*}{$\mathcal{S}$} & \multirow{2}{*}{\checkmark}  & \multirow{2}{*}{$(k, m)$} & \multirow{2}{*}{$T_{\text{Fit}}$} & \multicolumn{3}{c|}{Communication (MB)} & \multicolumn{3}{c}{Computation time (ms)} \\ \cline{6-11}
    &  &  &  &  & Base & \scriptsize{\sf{NFGen}} & \textbf{Save} & Base &\scriptsize{\sf{NFGen}} & \textbf{SpeedUp} \\[1pt] \hline\hline
    \multirow{6}*{\shortstack{$\textit{sigmoid}(x) = \frac{1}{1 + e^{-x}}$ \\  $x \in [-50, +50]$, $F(x) \in [0.0, 1.0]$ \\ Non-linear buildling-blocks: $2$}}  
    & A    & $\times$ &     (10, 8)     &  4.3    &     618 &       263 &        \textbf{60\%} &            147 &        23 &      \textbf{6.3$\times$} \\
    & B    & \checkmark &     (7, 10)     &  3.5    &     1   &         1 &        -5\% &            137 &       124 &      \textbf{1.1$\times$} \\
    & C    & \checkmark &     (5, 14)     &  3.5    &     4   &         4 &        -5\% &           1155 &       802 &      \textbf{1.4$\times$} \\
    & D    & \checkmark &     (5, 14)     &  3.5    &     18  &        19 &        -8\% &           1863 &      1525 &      \textbf{1.2$\times$} \\
    & E    & \checkmark &     (5, 14)     &  3.5    &     212 &       308 &       -45\% &          75949 &    106857 &      0.7$\times$ \\
    & F    & \checkmark &     (5, 14)     &  3.5    &     207 &       234 &       -13\% &           9732 &     11224 &      0.9$\times$ \\
    \midrule
    \multirow{6}*{\shortstack{$\textit{tanh}(x) = \frac{e^{x} - e^{-x}}{e^{x} + e^{-x}}$ \\  $x \in [-50, +50]$, $F(x) \in [-1.0, 1.0]$ \\ Non-linear buildling-blocks: $3$}}  
    & A    & $\times$ &     (9, 8)  &   4.5     &   1876    &       216 &       \textbf{90\%} &            335 &        21 &       \textbf{15.7$\times$} \\
    & B    & $\times$ &     (5, 9)  &   3.2     &   13      &         1 &       \textbf{92\%} &            800 &        80 &       \textbf{10.0$\times$} \\
    & C    & $\times$ &     (5, 9)  &   3.2     &   19      &         3 &       \textbf{83\%} &           5901 &       597 &       \textbf{9.9$\times$} \\
    & D    & $\times$ &     (5, 9)  &   3.2     &   64      &        14 &       \textbf{78\%} &           8882 &      1115 &        \textbf{8.0$\times$} \\
    & E    & $\times$ &     (5, 9)  &   3.2     &   996     &       197 &       \textbf{80\%} &         337530 &     68550 &        \textbf{4.9$\times$} \\
    & F    & $\times$ &     (5, 9)  &   3.2     &   966     &       150 &       \textbf{84\%} &          45486 &      7309 &        \textbf{6.2$\times$} \\
    \midrule
    \multirow{6}*{\shortstack{$\textit{soft\_sign}(x) = \frac{x}{1 + |x|}$ \\  $x \in [-50, 50]$, $F(x) \in [-1.0, 1.0]$ \\ Non-linear buildling-blocks: $2$}}  
    & A   & $\times$ &      (8, 8)  &  1.9      &       518 &       231 &         \textbf{60\%} &            131 &        21 &        \textbf{6.1$\times$} \\
    & B   & \checkmark &      NA      &  1.3      &       1   &         1 &         0\%  &             79 &        78 &        1.0$\times$ \\
    & C   & \checkmark &      NA      &  1.3      &       2   &         2 &         0\%  &            451 &       437 &        1.0$\times$ \\
    & D   & \checkmark &      NA      &  1.3      &        8  &         8 &         0\%  &            741 &       753 &        1.0$\times$ \\
    & E   & \checkmark &      NA      &  1.3      &       52  &        52 &         0\%  &          15507 &     15520 &        1.0$\times$ \\
    & F   & \checkmark &      NA      & 1.3       &       49  &        49 &         0\%  &           2315 &      2373 &        1.0$\times$ \\
    \midrule
    \multirow{6}*{\shortstack{$\textit{Normal\_dis}(x) = \frac{e^{-\frac{x^2}{2}}}{\sqrt{2\pi}}$ \\  $x \in [-10, +10]$, $F(x) \in [0.0, 0.4]$ \\ Non-linear buildling-blocks: $1$}}  
    & A    &  $\times$  &   (8, 12)     &   5.2  &      420     &       295 &        \textbf{30\%} &             67 &         24 &       \textbf{2.8$\times$}  \\
    & B    &  $\times$  &   (8, 12)     &   3.6  &      3       &         2 &        \textbf{45\%} &           4906 &       156 &       \textbf{31.5$\times$} \\
    & C    &  $\times$  &   (8, 12)     &   3.6  &      7       &         5 &        \textbf{27\%} &           5029 &       970 &        \textbf{5.2$\times$}  \\
    & D    &  $\times$  &   (8, 12)     &   3.6 &       24      &        23 &         \textbf{5\%} &           6588 &      1846 &        \textbf{3.6$\times$}  \\
    & E    &  $\times$  &   (5, 22)     &   3.6  &      257     &       481 &       -87\% &          89740 &    166328 &        0.5$\times$  \\
    & F    &  $\times$  &   (8, 12)     &   3.6  &      249     &       301 &       -21\% &          14908 &     14861 &        1.0$\times$  \\
    \midrule
    \multirow{6}*{\shortstack{$\textit{Bs\_dis}(x)\text{~\cite{birnbaum1969new}}= \left (\frac{\sqrt{x} + \sqrt{\frac{1} {x}}} {2\gamma x} \right) 
    \phi \left (\frac{\sqrt{x} - \sqrt{\frac{1} {x}}} {\gamma} \right)$ \\ $\gamma=0.5$, $x \in [10^{-6}, 30]$, $F(x) \in [0.0, 0.2]$ \\ Non-linear buildling-blocks: $3$}}  
    & A  & $\times$  &     (10, 8) &   4.0     &   2815 &      263     &       \textbf{90\%} &            630 &         22 &        \textbf{29.1$\times$}  \\
    & B  & $\times$   &     (7, 11) &   3.2     &   13 &         1      &        \textbf{89\%} &          11463 &       133 &       \textbf{86.1$\times$}  \\
    & C  & $\times$   &     (5, 16) &   3.2     &   23 &         5      &        \textbf{79\%} &          14631 &       915 &       \textbf{16.0$\times$}  \\
    & D  & $\times$   &     (5, 16) &   3.2     &   65 &        22      &        \textbf{66\%} &          19167 &      1763 &       \textbf{10.9$\times$}  \\
    & E  & $\times$   &     (5, 16) &   3.2     &   741 &       352     &        \textbf{53\%} &         239549 &    122325 &        \textbf{2.0$\times$}  \\
    & F  & $\times$   &     (5, 16) &   3.2     &   718 &       268     &        \textbf{63\%} &          42157 &     13136 &        \textbf{3.2$\times$}  \\
    \midrule  
    \end{tabular}
    \begin{tablenotes}
        \item * $T_{\text{Fit}}$ is the time for \pkm fitting in seconds. \checkmark indicates whether baseline achieves the accuracy requirements.
    \end{tablenotes}
\end{threeparttable}}
\end{table*}

We use 15 widely-used non-linear functions for performance evaluation, including 8 activation functions used in deep learning and 7 probability distribution functions. 
The input domain of each function is set to the interval without a close-to-zero derivative, as these intervals are hard to approximate while others can be simply approximated with some constants. 
We run these functions on all six protocols in the two MPC platforms, and compare the performance with \emph{direct evaluation} as the baseline  (except \textit{sigmoid}s which is the built-in functions in both platforms).  In all cases, we set the accuracy requirement to $\epsilon = 10^{-3}$ and $\hat{0} = 10^{-6}$ (defined in Eq.~\ref{eq:rela-dis}), and run the experiments on 10,000 evenly spaced $\hat{x}$ samples. 
We measure the computation time on 100-dimensional vectors.  Table~\ref{Tab:mb-activation} shows five examples and Appendix~\ref{sec:full-microbenchmark} lists all the 15 functions.
We have the following important observations:

\smallpara{Performance.  } 
1) \sysname achieves significant performance gain in 93\% of these cases, with an average speedup of $6.5\times$ and a max speedup of $86.1\times$ (\textit{Bs\_dis} on \repz). 

2) \sysname significantly reduces communication, with an average reduction of 39.3\% and a max of 93\%, but it is not proportional to the speedups.  This is because the OPPE algorithm uses a fixed number of communication rounds, and each round involves vectorized comparisons and multiplications (depending on $k$ and $m$), while the baseline evaluates the function step-by-step, and thus may involve more rounds, leading to longer computation time. 

3) The more complicated a function is, the more likely \sysname achieves a better speedup, for the same reason above.  

4) Smaller $m$ values perform better for most functions.  As both MP-SPDZ and PrivPy support vectorized \gtop operations, the latency is largely dependent of $m$, even on the 100-dimensional vectors. 

5) In \emph{malicious} protocols (settings E and F), \sysname is more likely to achieve smaller speedup or even slowdowns on a few functions like \textit{sigmoid}.  There are two reasons: a) the \mulop and \gtop are very slow, and the intensive \mulop's introduce more \emph{validation checks} (the \emph{Ps}-protocols use extra \emph{multiplication triplets} to detect cheating); b) the validation checks prevent efficient batch (vector) operations in MP-SPDZ, resulting in less efficient OPPE execution.

\smallpara{Effectiveness of the profiler.  } 
1) For the \textit{soft\_sign} function, on all protocols in MP-SPDZ, \sysname automatically falls back to direct evaluation (Section~\ref{sec:profiler}), while on PrivPy, it uses $\hat{p}_{8}^{8}$ polynomial and achieves a $6.1\times$ speedup.  This is because the function essentially computes an absolute value (equivalent to a \gtop, plus a \divop).  Both \gtop and \divop are significantly slower than \mulop in PrivPy, but it is not the case for MP-SPDZ, as Table~\ref{settings} shows.   Thus, based on the profiler results, \sysname chooses different evaluation strategies.  In fact, we manually test the polynomial approach for MP-SPDZ, and it is $1.6\times$ slower than direct evaluation (with the most efficient polynomial ($\hat{p}_{6}^{10}$)), showing the effectiveness of the profiler. 

2) The profiler can also select different $(k,m)$ settings for different protocols in MP-SPDZ.  For example, it chooses different $(k, m)$s for different protocols in computing \textit{sigmoid}.  Specifically, it uses a larger $k$ for \repz (more \mulop and fewer \gtop), as \gtop is $4\times$ slower than \mulop in \repz, according to Table~\ref{settings}. 

3) For \repf, \repzps, \repfps and \shamir, the profiler tends to choose larger $m$ for a smaller $k$, because the relative performance between \mulop and \gtop is about $1:1$ in these settings. Thus computing $O(\log{k})$ rounds of {\mulop}s is more expensive than a single round of vectorized \gtop.

\smallpara{Independence and parallelism.  }
We adopt two types of optimizations to accelerate the online phase performance:
1) Independent evaluation of $\times$ and $>$:  
As we have discussed in Section~\ref{sec:oppe}, it is independent to perform the comparisons and to compute $x$ to the $k$-th power. We evaluate the speedup on PrivPy by running them independently\footnote{MP-SPDZ does not support customized multi-threading in its user-level language so we do not adopt this optimization.}. Figure~\ref{Fig:task_parallel} shows the performance results with/without independent evaluation.  We can see the optimization provides a $1.3 - 1.4\times$ speedup even on small $100$-dimensional input.

2) Concurrency on large inputs:  Figure~\ref{Fig:multi-thread} shows the performance with varied numbers of threads on $10^{6}$ input. We observe good speedup up to $10$ threads (about $4\times$ speedup comparing with the single-thread) for all four functions we evaluate.  Using more than 10 threads decreases performance as given the vector size, cost of threading overweights the speedup. 

\begin{figure}[tb]
	\centering
	\subfigure[\rthread{Evaluate $\times$, $>$ Independently (IE)}]{ \label{Fig:task_parallel}
	\begin{minipage}{0.43\linewidth}
	\centering
	\includegraphics[width=\textwidth]{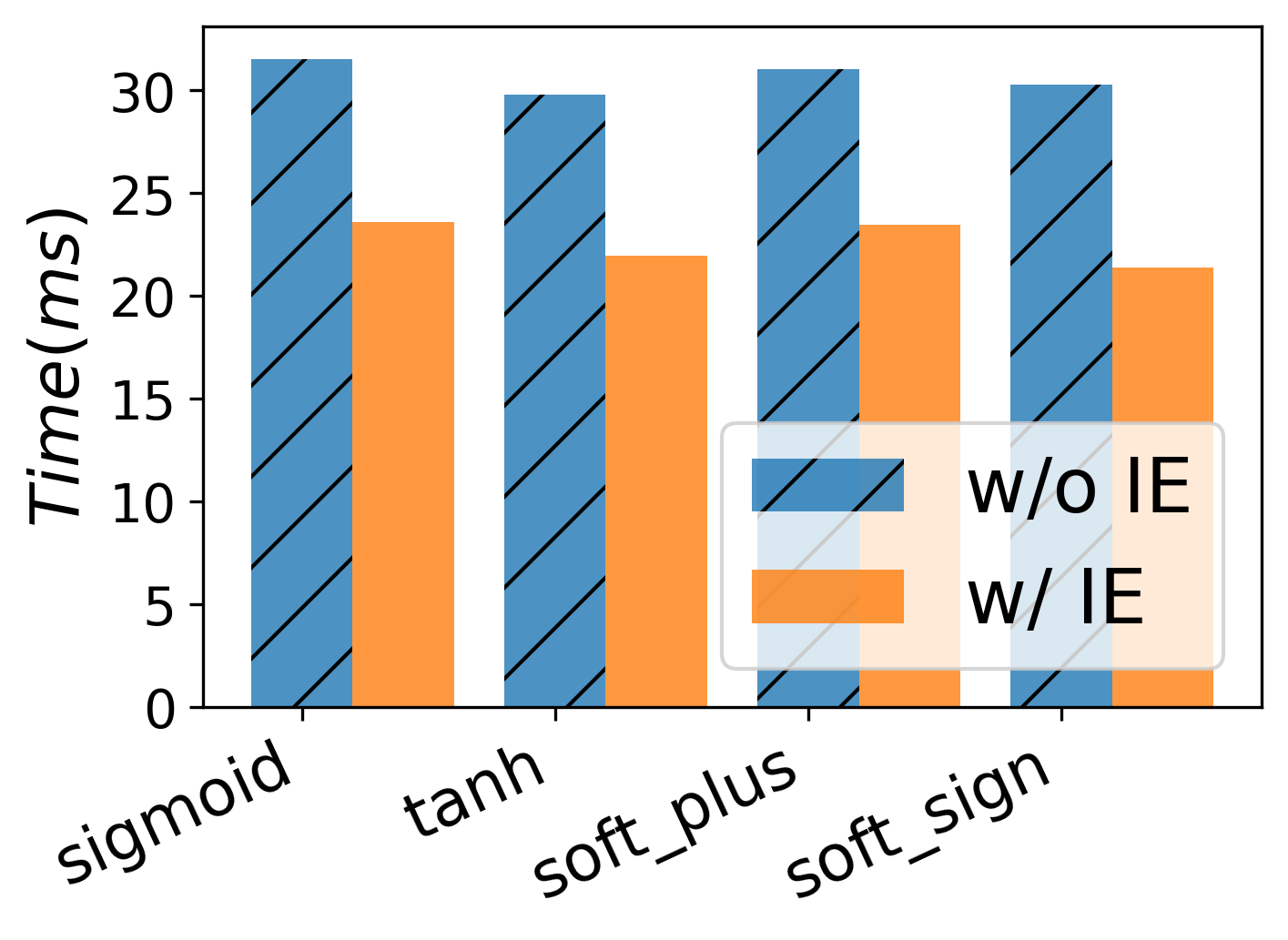} 
	\end{minipage}
	}
	\subfigure[\rthread{Multi-threads Acceleration}]{\label{Fig:multi-thread}
	\begin{minipage}{0.43\linewidth}
	\centering 
	\includegraphics[width=\textwidth]{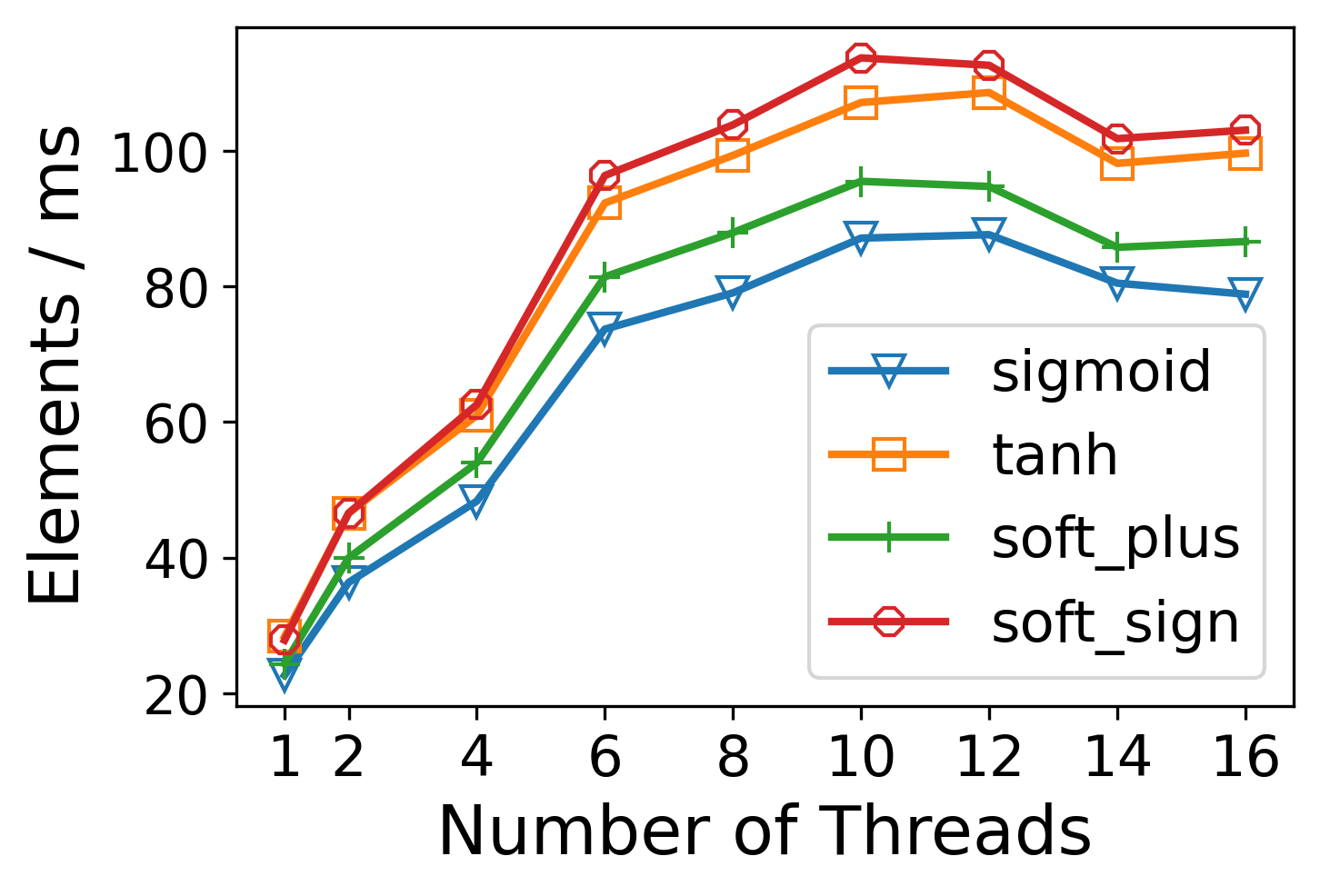}
	\end{minipage}
	}
	\vspace{-0.2in}
	\caption{\rthread{Performance Optimizations}}
	\label{Fig: parallel-optimizations}
\end{figure}

\subsection{Accuracy}
~\label{exp:accuracy}

\begin{figure}
	\flushleft 
	\subfigure{
		\includegraphics[width=0.45\textwidth]{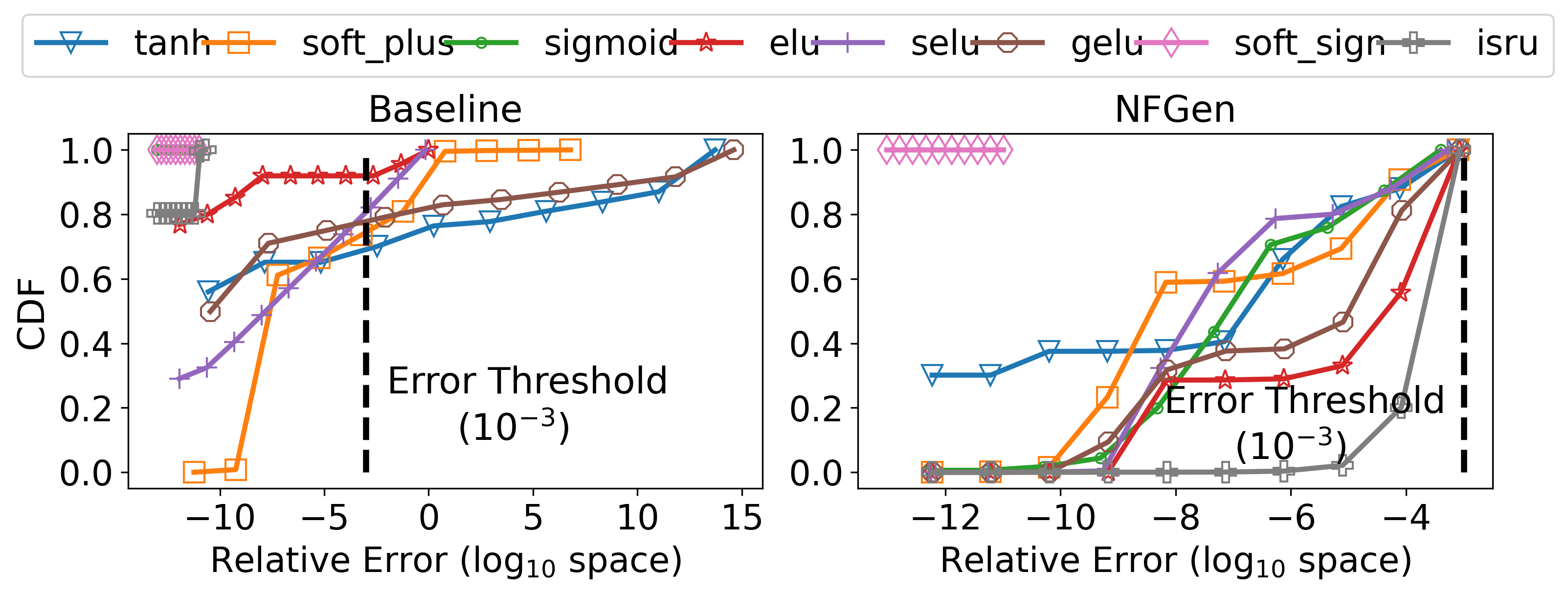}
		\vspace{-6mm}
	}
	\vspace{-4mm}
	\flushleft
	\subfigure{
		\includegraphics[width=0.45\textwidth]{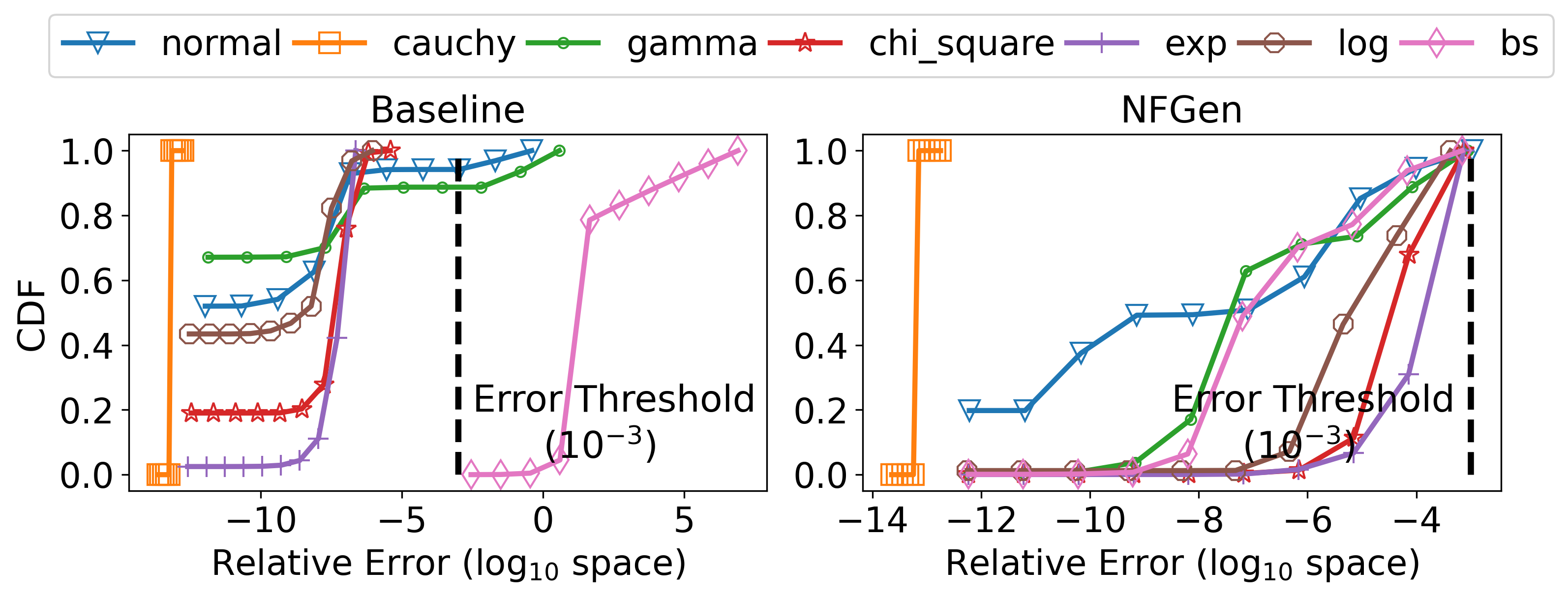}
	}
	\vspace{-0.23in}
	\caption{CDF of the Relative Errors}\label{Fig: accuracy-overview}
\end{figure}

\smallpara{Accuracy against the baseline.  } 
We compute the soft relative distance (SRD in Section~\ref{sec:overview}) on each of the 10,000 input samples.  We see a large variation of SRDs on different functions, and Figure~\ref{Fig: accuracy-overview} shows the \emph{cumulative distribution functions (CDFs)} of the SRDs on each function for both the baseline and \sysname.

We observe the following:
1) 100\% of the SRD of all 10k samples are under $10^{-3}$ (the target $\epsilon$), ranging between $10^{-12}$ and $10^{-3}$, as expected.  
2) In comparison, the baseline shows diverse errors across functions.  The $Cauchy\_dis$ is very accurate as it is just a \divop (\sysname falls back to direct evaluation in this case).
However, in some functions, the baseline shows errors way exceeding the $10^{-3}$ limit (dashed vertical line in Figure~\ref{Fig: accuracy-overview}) on over 10\% samples (Function $Bs\_dis$ has 100\% and $tanh$ has more than 35\%).
There are two reasons: precision loss due to concatenating multiple non-linear functions and errors due to overflow/underflow.
3) Baseline has a larger percentage of samples with smaller errors (e.g. $<10^{-9}$) than \sysname.  This is as expected too: \sysname is based on a regression model to approximate, while baseline is a direct computation.  However, we believe the predictable accuracy is more important than sometimes getting smaller errors.

\smallpara{Errors due to overflows/underflows.  } 
When the evaluation of \f uses an intermediate result in a large range, like \expop, it overflows given an $x$ with $|x| > \ln{2^{(96-48-1)}} \approx 33.2$ in MP-SPDZ with $\langle 96,48 \rangle$-FXP. 
The top two rows in Figure~\ref{Fig:problem-examples} plot SRD vs. $x$ value for the baselines, and show the overflow cases, where the SRD can exceed $10^{5}$, even if the range of \f is perfectly representable.
In comparison, the bottom two rows in Figure~\ref{Fig:problem-examples} show that SRDs are less than $10^{-3}$ in \sysname.  It also uncovers that \sysname has larger SRD when \f is close to zero, as the accuracy of polynomial approximation are constrained by the limited resolution of FXP.

\begin{figure}
	\flushleft
	\subfigure{
		\includegraphics[width=0.45\textwidth]{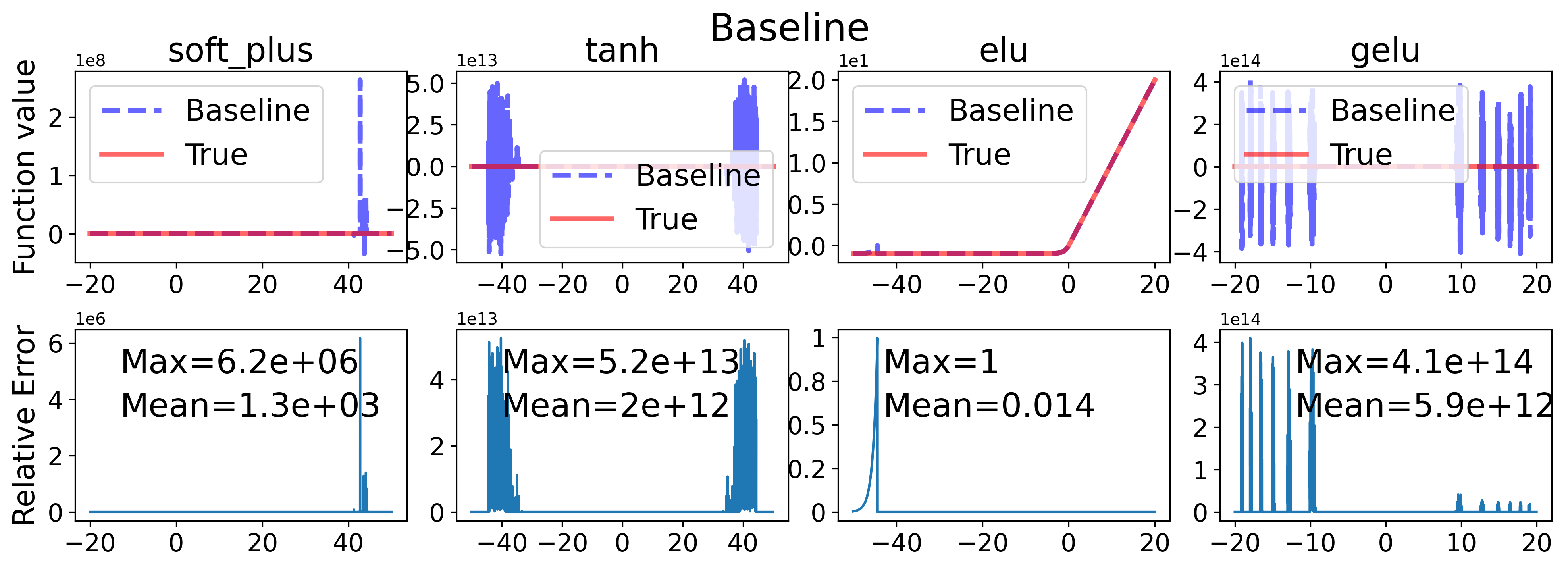}
		\vspace{-6mm}
	}
	\vspace{-4mm}
	\flushleft
	\subfigure{
		\includegraphics[width=0.45\textwidth]{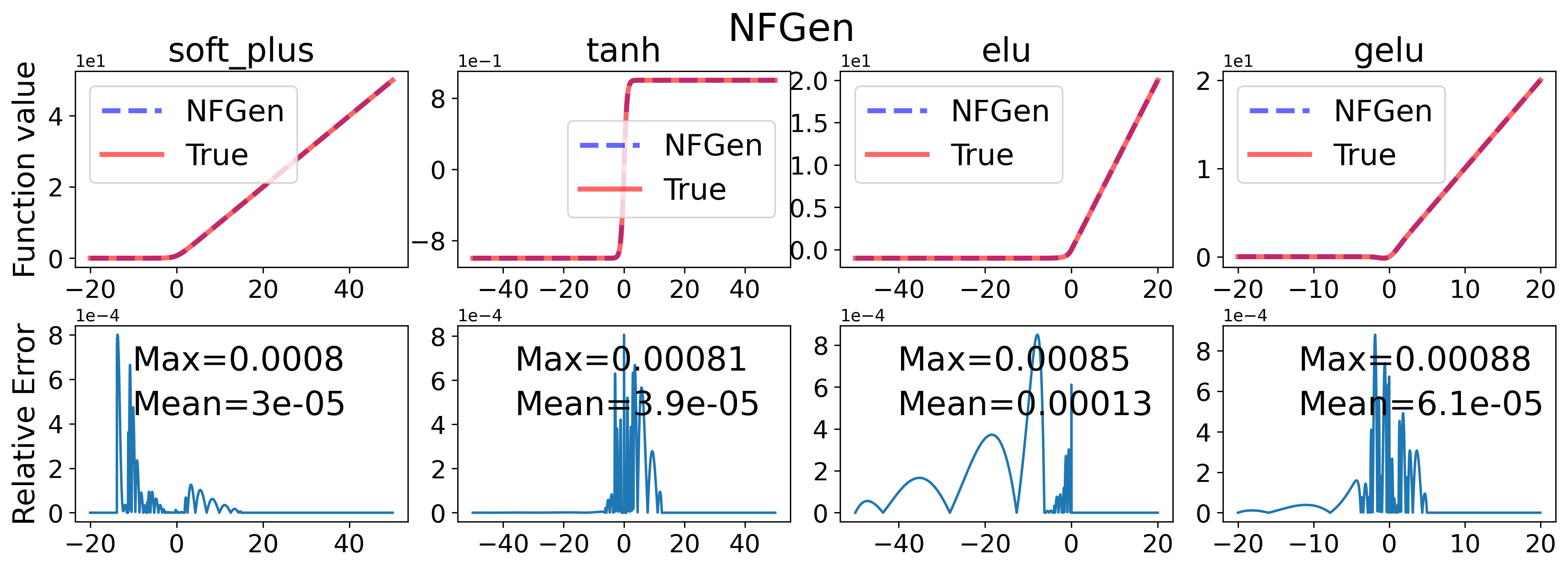}
		\vspace{-6mm}
	}
	\flushleft
	\vspace{-0.23in}
	\caption{Overflow Error Examples}\label{Fig:problem-examples}
\end{figure}

\smallpara{Error accumulation.  } 
The error accumulation problem widely exists in scientific computing~\cite{atkinson2005theoretical}, even with double precision numbers. Using FXP makes the problem worse, especially when the calculation of the target functions has many steps.  As \sysname approximate the entire \f in one shot, it does not accumulate errors. 

We take the $Bs\_dis$ function (Row 5 in Table~\ref{Tab:mb-activation}), whose calculation has four steps, as an example. 1) $x_{11} \leftarrow \sqrt{x}; \, x_{12} \leftarrow \frac{1}{x}$; 2) $x_{2} \leftarrow \sqrt{x_{12}}$; 3) $x_{3} \leftarrow \phi(\frac{x_{11} - x_{2}}{\gamma})$ and 4) $x_{4} \leftarrow (\frac{x_{11} + x_{2}}{2\gamma * x}) * x_{3}$.
Using 10k $\hat{x}$ samples, we compute the SRD after each step, and plot the CDF in Figure~\ref{Fig: error-accumulation-case} (left).  We can see that although the first two steps results in negligible SRD that is smaller than $10^{-10}$, more samples starts to show larger SRD after steps 3 and 4.  After step 4, only less than half of the samples meet the accuracy requirement of $10^{-3}$.  Figure~\ref{Fig: error-accumulation-case} (right) shows the evaluation results, and we find obvious inaccuracies for $\hat{x} \in [0.054, 18]$, without overflow or underflow.  In comparison, \sysname successfully limits the error below $3.8\times 10^{-4}$ using $\hat{p}_{5}^{11}$.  

\begin{figure}
    \centering
    \includegraphics[width=0.45\textwidth]{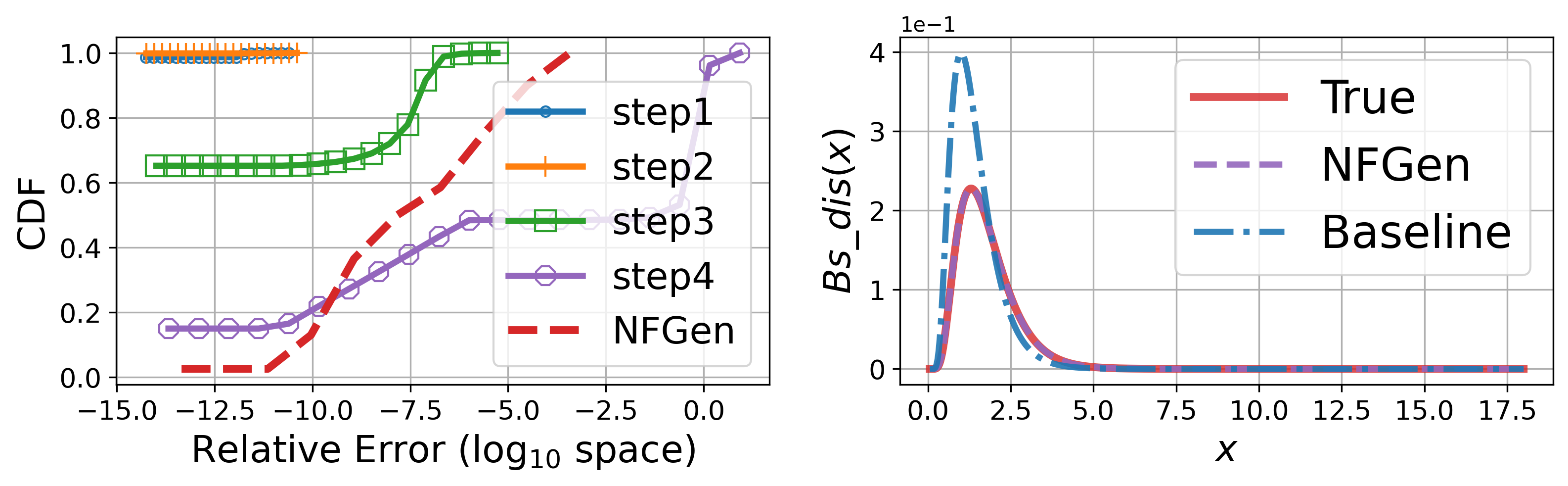}
	\vspace{-0.2in}
    \caption{Error Accumulation Case}\label{Fig: error-accumulation-case}
\end{figure}

\smallpara{Accuracy loss in secret sharing.  } 
The secret sharing and reconstruction processes in MPC platforms may introduce extra inaccuracies as they approximate each real value in fixed-point shares
~\cite{catrina2010secure, mohassel2017secureml} (we call it \emph{secret sharing reconstruction error (SSRE)}).
For example, we observe that the value of exactly $1\times 10^{-14}$, after secret sharing and reconstruction, may become $1.06581\times10^{-14}$.  This SSRE also adds inaccuracy to the final evaluation result.  To quantify the contribution of SSRE, we compare the SRD using real MPC to our simulated FXP (without SSRE).   Table~\ref{error_comparision} shows the comparison (* highlights cases with different results).  We can see that SSRE does increase the inaccuracy, but the contribution is small comparing to other sources of inaccuracy. 

\begin{table}[tb]
	\caption{Secret Sharing Reconstruction Loss}
    \label{error_comparision}
	\vspace{-0.2in}
	\footnotesize
    \centering
    \begin{threeparttable}
    \begin{tabular}{c|cc|ccc}
    \midrule[1.1pt]
    \multirow{2}{*}{} & \multicolumn{2}{c|}{\textbf{Simulation }} & \multicolumn{2}{c}{\textbf{Secret Sharing}} &  \\\hline
     & Max($\times 10^{-4}$) & Mean($\times 10^{-5}$) & Max($\times 10^{-4}$) & Mean($\times 10^{-5}$) &  \\\hline\hline
    tanh & $8.06$ & $3.88$ & $8.06$ & $3.88$ &  \\
    soft\_plus$^{*}$ & $7.67$ & $2.94$ & $8.00$ & $2.98$ &  \\
    sigmoid & $3.98$ & $1.87$ & $3.98$ & $1.87$ &  \\
    elu & $8.50$ & $12.87$ & $8.50$ & $12.87$ &  \\
    selu & $5.24$ & $1.94$ & $5.24$ & $1.94$ &  \\
    gelu & $8.79$ & $6.11$ & $8.79$ & $6.11$ &  \\
    isru & $8.61$ & $28.18$ & $8.61$ & $28.18$ &  \\
    normal\_dis$^{*}$ & $3.40$ & $1.01$ & $10.49$ & $2.53$ &  \\
    gamms\_dis & $8.80$ & $3.11$ & $8.80$ & $3.11$ &  \\
    chi\_square & $7.27$ & $10.70$ & $7.50$ & $10.70$ &  \\
    exp\_dis & $7.39$ & $19.76$ & $7.39$ & $19.77$ &  \\
    log\_dis & $4.24$ & $3.28$ & $4.24$ & $3.28$ & \\
    bs\_dis$^{*}$ & $5.96$ & $1.98$ & $6.91$ & $1.97$ & \\\hline
    \end{tabular}
\end{threeparttable}
\end{table}

\smallpara{Accuracy with different FXP widths.  }
Many MPC platforms offer configurable $\langle n, f\rangle$ for FXP numbers.  While reducing $n$ saves computation cost, a small $n$ limits the input domain for both baseline methods and \sysname, because we need to represent all inputs, intermediate results and outputs with the number of bits.  
We conduct experiments with $n$ ranging from 32 to 128 and compare the supported $\hat{x}$ domain between the baseline and \sysname and summarize the results in Table~\ref{Tab:domain restriction}. 

We observe: 1) As $n$ gets smaller, the ranges shrink for both cases.  \sysname supports a much larger domain even for $n=32$.  This is because the baseline domain is severely limited by the range of intermediate results, e.g. \divop overflows when $x$ gets close to zero and \expop overflows when $x> \ln{2^{n-f-1}}$.  
2) In contrast, $n$ affects the domain in \sysname in two different ways:  a) it directly limits representable $x$ to $|2^{n-f-1}|$; b) it limits the representation of $x^k$, forcing us to use only small $k$ values (Column 3).  We see that \sysname can automatically adapt to the $n$ settings by reducing the $k$ values to prevent overflows. 
3) As expected, a small $k$ limits the accuracy we can achieve.  We empirically determine the minimal possible accuracy (both $\hat{0}$ and $\epsilon$) when we require $m<1000$ (Column 1), and find that even in the $\langle 32, 16\rangle$ setting (i.e. the max representable number is only $2^{15}$), we can still maintain an $\epsilon \approx 5\%$, covering almost the entire representable domain between $-10^{4}$ and $10^{4}$.

    \begin{table}
        \caption{Domain Restriction with Data Representation}\label{Tab:domain restriction}
        \vspace{-0.2in}
        \footnotesize
        \begin{threeparttable}
        \begin{tabular}{c|c|c|c|c}
            \midrule[1.1pt]
            Config & \f & $(k, m)$ & Origin $D$ & Ours $D$ \\
            \hline \hline
            \multirow{4}{*}{\shortstack{$\langle 128, 64 \rangle$ \\ $\hat{0}=10^{-6}$ \\ $\epsilon=10^{-4}$}} 
            & $tanh$ & $(6, 11)$ & $[-44.4, 44.4]$ & $(-10^{19}, 10^{19})$ \\ 
            ~ & \textit{soft\_plus} & $(9, 9)$ & $[-44.4, 44.4]$ & $(-10^{19}, 10^{19})$ \\
            ~ & \textit{Normal\_dis} & $(6, 48)$ & $[-9.4, 9.4]$ & $(-10^{19}, 10^{19})$ \\
            ~ & \textit{Bs\_dis} & $(10, 9)$ & $[0.0, 24.1]$  &  $(0, 10^{19})$ \\
            \midrule
            \multirow{4}{*}{\shortstack{$\langle 96, 48 \rangle$ \\ $\hat{0} = 10^{-6}$ \\ $\epsilon=10^{-3}$}} 
            & $tanh$ & $(6, 11)$ & $[-33.3, 33.3]$ & $(-10^{14}, 10^{14})$ \\ 
            ~ & \textit{soft\_plus} & $(6, 11)$ & $[-33.3, 33.3]$ & $(-10^{14}, 10^{14})$ \\
            ~ & \textit{Normal\_dis} & $(10, 9)$ & $[-8.2, 8.2]$ & $(-10^{14}, 10^{14})$ \\
            ~ & \textit{Bs\_dis} & $(5, 14)$ & $[0.1, 18.6]$  &  $(0, 10^{14})$ \\
            \midrule
            \multirow{4}{*}{\shortstack{$\langle 64, 32 \rangle$ \\ $\hat{0} = 10^{-5}$ \\ $\epsilon=10^{-3}$}} 
            & $tanh$ & $(5, 9)$ & $[-22.2, 22.2]$ & $(-10^{9}, 10^{9})$ \\ 
            ~ & \textit{soft\_plus} & $(7, 9)$ & $[-22.2, 22.2]$ & $(-10^{9}, 10^{9})$ \\
            ~ & \textit{Normal\_dis} & $(10, 9)$ & $[-6.7, 6.7]$ & $(-10^{9}, 10^{9})$ \\
            ~ & \textit{Bs\_dis} & $(5, 14)$ & $[0.1, 13.0]$  &  $(0, 10^{9})$ \\
            \midrule
            \multirow{4}{*}{\shortstack{$\langle 32, 16 \rangle$ \\ $\hat{0} = 10^{-2}$ \\ $\epsilon=5\cdot 10^{-2}$}} 
            & $tanh$ & $(4, 6)$ & $[-11.1, 11.1]$ & $(-10^{4}, 10^{4})$ \\ 
            ~ & \textit{soft\_plus} & $(4, 6)$ & $[-11.1, 11.1]$ & $(-10^{4}, 10^{4})$ \\
            ~ & \textit{Normal\_dis} & $(4, 7)$ & $[-4.7, 4.7]$ & $(-10^{4}, 10^{4})$ \\
            ~ & \textit{Bs\_dis} & $(5, 5)$ & $[0.1, 7.4]$  &  $(0, 10^{4})$ \\
            \midrule
        \end{tabular}
        \begin{tablenotes}
            \item * The function range are actually defined by $\hat{0}$, we set outer range to default constant (e.g., $tanh(\xhat) = -1$ for $\forall \xhat \leq -18.79$).
        \end{tablenotes}
        \end{threeparttable}
    \end{table}

\smallpara{Limitations on accuracy.  }  As \sysname uses approximations, there is no guarantee that it will find a workable polynomial with small $\epsilon$s.  We have shown that when $\epsilon = 10^{-3}$, we can successfully find approximations for all 15 functions.  When we set $\epsilon = 10^{-4}$, we fail to find a good \pkm for $Gamma\_dis$. If we further limit $\epsilon$ to $10^{-5}$, six out of the 15 functions fails to fit.  However, considering that MPC is mostly employed on data mining applications that do not require high precision, we believe \sysname strikes the right balance between efficiency and predictable accuracy. 

\subsection{Applied in Real Algorithms}
\label{exp:real-demo}

\sysname benefits MPC algorithms mainly in two ways.  First, it improves both performance and accuracy for existing MPC algorithms.  Second, it allows people to evaluate advanced non-linear functions that we cannot construct with simple built-ins.  We use logistic regression (LR) as an example to show the first benefit, and use a series of special functions and the $\chi^{2}$ test to show the second.

\smallpara{Logistic Regression (LR) Accuracy.  }
LR~\cite{bishop2006pattern} is one of the most utilized data mining algorithms, both in plaintext and MPCs.  
The major challenge for MPC is the slow performance of evaluating \textit{sigmoid}.  Prior projects use a $3$-piece linear function~\cite{mohassel2017secureml, mohassel2018aby3}, and ~\cite{hesamifard2018privacy} uses single \textit{Chebyshev} polynomial to approximate the \textit{sigmoid}.  Figure~\ref{Fig: sigmoid-cases} compares the \textit{sigmoid} function with different approximations. 

\begin{figure}[tb]
	\centering
	\includegraphics[width=0.3\textwidth]{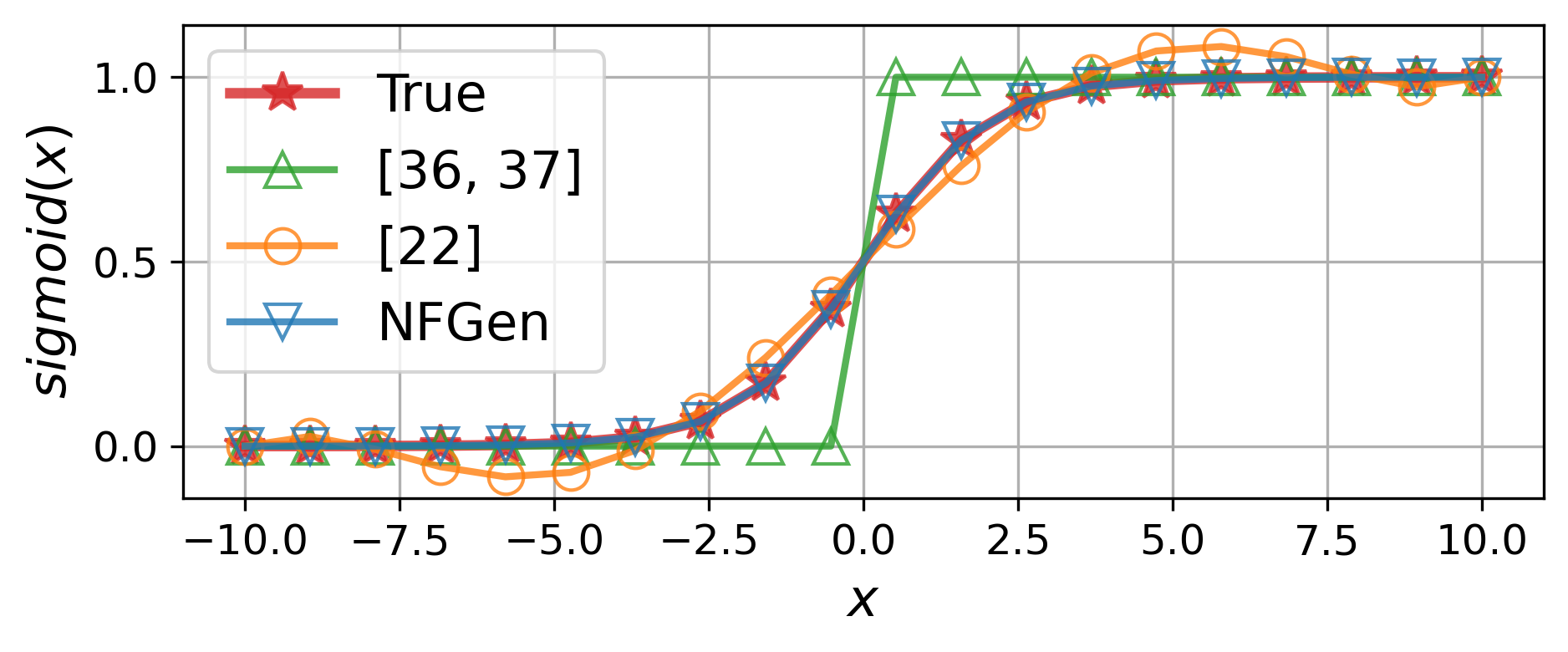}
	\vspace{-0.2in}
	\caption{Sigmoid Approximation Comparison}\label{Fig: sigmoid-cases}
\end{figure}

People argue that the accuracy of \textit{sigmoid} does not affect the LR accuracy\cite{mohassel2017secureml, mohassel2018aby3, hesamifard2018privacy}.  To evaluate this argument, we generate four datasets using Python \code{sklearn}'s \code{make\_classifiation()} method, and use them to train LR models using different approximations. Table~\ref{Tab:sigmoid-comp} reports the LR prediction accuracy.  We see that the $3$-piece approximation can lead to significant LR accuracy loss, while the approximation in~\cite{hesamifard2018privacy} slightly reduces accuracy.  In comparison, \sysname achieves almost the same accuracy as the plaintext result.  Thus,  \textit{sigmoid} accuracy does affect LR performance in some cases.  \sysname provides an efficient way to evaluate \textit{sigmoid} with high accuracy, eliminating the need for \emph{ad hoc} approximations. 

\begin{table}[tb]
	\caption{Performance Analysis of \sigmoid}\label{Tab:sigmoid-comp}
	\vspace{-0.2in}
	\footnotesize
	\centering
	\begin{threeparttable}
		\setlength{\tabcolsep}{2.5mm}
		\begin{tabular}{c|cccc}
			\midrule[1.2pt]
			Data setting$^{*}$ & Real & \sysname & SecureML~\cite{mohassel2017secureml,mohassel2018aby3} & \cite{hesamifard2018privacy} \\
			\hline\hline
			$(0.5, 3, 0.2)$ & 59.2 & 59.2 & 49.2 & 58.6 \\
			$(0.6, 3, 0.1)$ & 62.5 & 62.5 & 50.8 & 62.2 \\
			$(0.7, 3, 0.1)$ & 65.8 & 65.8 & 51.4 & 65.3 \\
			$(0.7, 4, 0.1)$ & 61.9 & 61.9 & 50.0 & 61.5 \\
			\midrule
		\end{tabular}
		\footnotesize
		\begin{tablenotes}
			\item * The three numbers are arguments \smallcode{class\_sep}, \smallcode{clusters\_per\_class} and \smallcode{learning\_rate} passed to the Python \smallcode{sklearn} library's \smallcode{make\_classifiation()} funtion to generate the dataset.
		\end{tablenotes}
	\end{threeparttable}
\end{table}

\smallpara{LR performance.  }
We use 3 real datasets to evaluate LR training and inference time on PrivPy.  We omit evaluation on MP-SPDZ as it needs to pre-compile all input data into the program, but the compiler fails on large datasets.  Independent of the dataset, we set the $\hat{x}$ domain to $[-10, 10]$, which is a typical setting in practice when the distribution of dataset is unknown (when $x \notin [-10, 10]$, output $0$ or $1$).
Table~\ref{Tab:end2end-LR} shows the results.  We can see that \sysname achieves $3.5\times$ to $9.5\times$ speedup in training and $1.8\times$ to $2.3\times$ speedup for inference using $\hat{p}_{8}^{6}$ and with same Accuracy as plaintext LR.

\begin{table}[tb]
	\caption{Logistic Regression Speedups}\label{Tab:end2end-LR}
	\vspace{-0.2in}
	\footnotesize
	\centering
	\begin{threeparttable}
		\begin{tabular}{c|ccc}
			\midrule[1.2pt]
			Dataset &  Method & Train(sec) & Test(sec) \\[2pt]
			\hline \hline
			\multirow{2}*{\shortstack{~ \\[2pt] Adult~\cite{kohavi1996scaling} \\[2pt]($48,842 \times 65$)}} & 
			PrivPy & 413.1 & 1.8 \\[3pt]
			~ & \sysname  & 43.6 / $9.5\times$ & 0.8 / $2.3\times$ \\[3pt]
			\midrule
			\multirow{2}*{\shortstack{~ \\[2pt] Bank~\cite{moro2014data} \\ ($41,188 \times 63$)}}  & 
			PrivPy & 72.8 & 1.6 \\[3pt]
			~  & \sysname  & 20.4 / $3.6\times$ & 0.8 / $2.0\times$ \\[3pt]
			\midrule
			\multirow{2}*{\shortstack{~ \\[2pt] Branch~\cite{BranchTaken} \\($400,000 \times 480$)}} & 
			PrivPy  & 703.8 & 12.2 \\[3pt]
			~ & \sysname  & 199.9 / $3.5\times$ & 6.9 / $1.8\times$ \\[3pt]
			\midrule
		\end{tabular}
	\end{threeparttable}
\end{table}

\begin{table}[tb]
    \caption{Special Functions Demonstration}
    \label{Tab:sp-funcs}
	\vspace{-0.2in}
	\footnotesize
    \centering
    \begin{threeparttable}
    \begin{tabular}{c| c c c}
    \midrule[1.2pt]
    Target Function & Parameter  & $(k, m)$ & $T_{\text{Fit}}$(sec)\\[2pt]
    \hline \hline
    \multirow{3}{*}{$\gamma(x, z) = \int_{0}^{x}t^{z-1}e^{t} \,dt$, $x \in [0, 15]$} 
    & $z=1$ & (6, 4) & 1.1	\\
    & $z=2$ & (5, 6) & 1.6	\\
    & $z=3$ & (6, 6) & 2.0	\\
    \midrule
    \multirow{3}{*}{$\Gamma(x, z) = \int_{x}^{\infty}t^{z-1}e^{t}\,dt$, $x \in [0, 10]$}
    & $z=1$ & (6, 6) & 1.1 \\
    & $z=2$ & (8, 4) & 1.1	\\
    & $z=3$ & (7, 4) & 1.2	\\
    \midrule
    $erf(x) = \frac{2}{\sqrt{\pi}}\int_{0}^{x}e^{-t^2}\,dt$, $x \in [0, 5]$ & NA  & (4, 6) & 0.8 \\
    \midrule
    $\Phi(x) = \frac{2}{\sqrt{2\pi}}\int_{0}^{x}e^{\frac{-t^2}{2}}\,dt$, $x \in [-5, 5]$ &  NA  & (8, 6) & 1.2\\
    \midrule
    \end{tabular}
    \end{threeparttable}
\end{table}

\smallpara{Hard-to-implement functions.  }
A big problem MPC practitioners face is how to implement some commonly-used but hard-to-implement functions, such as $\gamma{(x)}$, $\Gamma{(x)}$ and $\Phi{(x)}$.  These functions are defined as integrals, and it takes much mathematical skills to approximate them using the limited operators in MPC (and impossible sometimes).  \sysname naturally solves the problem for all Lipschitz continuous functions as long as there is a plaintext implementation available.
We demonstrate $8$ hard-to-implement functions in Table~\ref{Tab:sp-funcs}.  We see that like other functions, it only takes about $1$-$2$ seconds to generate the approximation and achieve small $k$ and $m$ meeting the same accuracy requirement.

\smallpara{The $\chi^{2}$ test on real datasets. }
\textit{$\chi^{2}$ test}~\cite{mendenhall2012introduction} is a classic statistical method.  Unfortunately, the $p$ value from $\chi^2$ test depends on the $\gamma$ and $\Gamma$ functions in Table~\ref{Tab:sp-funcs}, as $ p = 1 - \frac{\gamma(\frac{k}{2}, \frac{x}{2})}{\Gamma(\frac{k}{2})}$,
where $\frac{\gamma(\frac{k}{2}, \frac{x}{2})}{\Gamma(\frac{k}{2})}$ is the CDF of $\chi^2$ distribution,  $x$ is the statistical value and $k$ is the degree-of-freedom (Dof) parameter that is typically set to the number of classes minus 1.  No current MPC framework supports $\chi^{2}$ test yet, to our knowledge.  
We show that we can easily implement $\chi^{2}$ test with \sysname on real datasets.  The datasets contain features of a patient with certain diseases, and as a typical task in medical research, we use $\chi^2$ test to determine whether the probability of a disease is correlated to a feature.  Table~\ref{tab:chi-1} shows that we can evaluate cases with different Dofs, and achieve the same result as in plaintext (with 3 significant digits).

\begin{table}[tb]
	\caption{$\chi^{2}$ Test using Real Datasets}
	\label{tab:chi-1}
	\vspace{-0.2in}
	\footnotesize
	\centering
	\begin{threeparttable}
		\begin{tabular}{c|c|cccc}
			\midrule[1.2pt]
			Dataset & Feature & Dof  & $(k,m)$ & Error & Time(sec) \\[2pt]
			\hline \hline
			\multirow{6}*{\shortstack{Cervical~\cite{machmud2016behavior} \\[2pt]($72 \times 20$) \\ $5$ features for demo}} &
			Sexual behavior & 5 & $(6, 11)$ & 0 &  \multirow{6}{*}{13.8} \\
			~ & Eating behavior & 7 & $(5, 13)$ & 0  & \\
			~ & Personal hygine & 11 & $(5, 14)$ & 0  & \\
			~ & Social support & 11 & $(5, 14)$ & 0  & \\
			~ & Attitude & 6 & $(5, 11)$ & 0  & \\
			\midrule
			\multirow{3}*{\shortstack{Sepsis~\cite{chicco2020survival} \\[2pt]($110,204 \times 3$)}} &
			age & 10 & $(5, 12)$ & 0  & \multirow{3}{*}{63.3}\\
			~ & sexual & 1 & $(5, 98)$ & 0  & \\
			~ & episode number & 4 & $(5, 9)$ & 0 & \\
			\midrule
		\end{tabular}
	\end{threeparttable}
\end{table}

\subsection{Effectiveness of Design Choices}
\label{sec:ablation-study}

We conduct an \emph{ablation study} to evaluate the design choices in \sysname, including the \emph{Scaling factor} (Algorithm~\ref{Algo: scalepoly}), \emph{Residual boosting} (Algorithm~\ref{residual-boosting}) and the \emph{Merge stage} in Algorithm~\ref{Algo:sampledpoly}, using the same setting as in Section~\ref{micro-benchmark}, on \repz.

	\begin{table}[tb]
		\caption{Ablation Studies}\label{Tab:ablation-study}
		\vspace{-0.2in}
		\footnotesize
		\begin{threeparttable}
			\centering
			\begin{tabular}{c|c|c|c|c|c}
				\midrule[1.1pt]
				\f & Metric & \sysname & no merge & no boosting & no scaling \\ \hline\hline
				\multirow{4}*{$sigmoid$} 
				& $T_{\text{Fit}}$ & 3.1 & 2.2 & 14.0 & 2.6\\
				& Best $(k, m)$ & (7, 10) & (5, 20) & (7, 10) & (7, 12) \\
				& Failures & 0 & 0 & 1 & 0\\
				\midrule
				\multirow{4}*{\textit{soft\_plus}} 
				& $T_{\text{Fit}}$ & 2.7 & 1.8 & 14.1 & 2.0 \\
				& Best $(k, m)$ & (7, 4) & (4, 30) & (7, 4) & (4, 9)\\
				& Failures & 0 & 0 & 1 & 0 \\
				\midrule
				\multirow{4}*{$selu$} 
				& $T_{\text{Fit}}$ & 1.3 & 0.9 & 1.2 & 36.0 \\
				& Best $(k, m)$ & (8, 9) & (6, 14) & (8, 9) & (6, 12)\\
				& Failures & 0 & 0 & 0 & 3\\
				\midrule
				\multirow{4}*{\textit{Gamma\_dis}} 
				& $T_{\text{Fit}}$ & 4.3 & 1.1 & 40.8 & 1.2 \\
				& Best $(k, m)$ & (7, 21) & (5, 35) & NA & (5, 26) \\
				& Failures & 0 & 0 & 7 & 0 \\
				\midrule
			\end{tabular}
		\end{threeparttable}
	\end{table}

We evaluate each technique on four functions, and Table~\ref{Tab:ablation-study} summarizes the result.  We measure three metrics on each ablation case: 1) $T_{\text{Fit}}$ is the runtime of Algorithm~\ref{Algo:sampledpoly}; 2) Best $(k,m)$ is the best-performance \pkm we find; 3) Failures count the number of cases where we do not find a feasible \pkm for $k$ with $m<m_{max}$.

Key observations include:
1) Though each technique takes extra pre-computation time, the overhead is less than 1 second per technique.  However, without them some functions are even slower to fit, e.g. $selu$ w/o scaling and $Gamma\_dis$ w/o residual boosting are both slower because some otherwise possible \pkm will not pass the check and results in more searching steps. 
2) The merge step significantly reduces $m$, as the splitting strategy often unnecessarily increases the number of pieces.  
3) Failures are more often if we remove residual boosting or scaling, showing that they effectively make some approximation possible by remedying inaccuracies introduced by the FLP-FXP conversion.

\section{Conclusion and Future Work}
\label{sec:conclusion}

Creating general-purpose MPC platforms is analogous to creating a new computation system from scratch using MPC primitives instead of instructions. Non-linear function evaluation, akin to plaintext numeric libraries, is one of these systems' \emph{foundations}. Prior approaches either naively attempted to reuse plaintext algorithms that resulted in erroneous results and/or slow performance, or developed \emph{ad hoc} approximations that were tightly coupled with either specific functions or MPC platforms. Neither method possesses the generality or performance necessary to serve as a viable \emph{foundation}. \sysname is, to our knowledge, the first attempt for a generic solution. We can accurately approximate general non-linear functions using piecewise polynomials by properly handling FXP and FLP operations. We achieve portability across multiple MPC systems and protocols by utilizing code generation and profiler-based performance prediction. Extensive evaluations verify our approach's effectiveness, accuracy, and generality.

As future work, we are going to support more MPC platforms, explore approximation algorithms with stronger theoretical guarantees, as well as support the evaluation of multi-dimensional non-linear functions.

\section*{Acknowledgements}
  We thank Menghua Cao and Zhilong Chen for insightful discussions during the design of \sysname. We thank Yuanxi Dai and Xinze Li for their assistance with the security analysis. We thank Xiang Wang and Haoqing He for their help during the implementations.

  This work is supported in part by the National Natural Science Foundation of China (NSFC) Grant 71872094 and gift funds from Nanjing Turing AI Institute.

\bibliographystyle{ACM-Reference-Format}
\bibliography{main}

\appendix

\section{Analysis of Algorithm}
\subsection{Maximum SRD Analysis for Section~\ref{sec:overview}}\label{appendix:precision}
    We prove the upper-bound of the real maximum \emph{soft relative distance} between the piecewise polynomial \pkm and the target Lipschitz continuous \f over domain $\xhat \in [a, b]$ in this section. 

    \begin{theorem}[Upper-bound of maximum soft relative distance between approximation and target function]\label{theorem:precision-bound}
        By constraining the maximum \emph{soft relative distance} (Eq~\ref{eq:rela-dis}) over sample set $\mathbf{\xhat} = \{\xhat_1, ..., \xhat_i, ..., \xhat_{N}\}$ as Eq~\ref{cons: accuracy}, where the sampling interval $r = \frac{b-a}{N}$, there exist constant $C$, such that the piecewise polynomial \pkm and target Lipschitz continuous function \f satisfy:
        \begin{equation*}
            \max_{\xhat\in [a, b]}|F(\xhat) - \text{\pkm}(\xhat)|_d \leq 
            \begin{cases}
                \frac{C \cdot r}{|F(\xhat)|} + 1, & |F(\xhat)| > \hat{0} \\
                C \cdot r + \epsilon, & |F(\xhat)| \leq \hat{0}
            \end{cases}
        \end{equation*}
    \end{theorem}

    \begin{proof}
        As \f and piecewise polynomial \pkm$(x)$ (refer to $p(x)$ in the following proof for simplicity) are both Lipschitz continuous, there exist Lipschitz constants $C_F$ and $C_{p}$ that for any interval $[\xhat_i, \xhat_{i+1}]$, where $\xhat_i \, \text{and}\, \xhat_{i+1}\, \in \mathbf{\xhat}$, $|\xhat_{i+1} - \xhat_{i}| = r$ are successive sample points, we have:

        \begin{equation}\label{eq:L-bound}
            \begin{aligned}
                & |F(\xhat) - F(\xhat_i)| \leq C_F \cdot |\xhat - \xhat_i| \leq C_F \cdot r, \\
                & |p(\xhat) - p(\xhat_i)| \leq C_{p} \cdot |\xhat - \xhat_i| \leq C_{p} \cdot r,
            \end{aligned}
        \end{equation}
        for $\forall \xhat \in [\xhat_i, \xhat_{i+1}]$.

        As the accuracy constraint (Eq~\ref{cons: accuracy}) satisfies on sample set $\mathbf{\xhat}$, we have:  
        \begin{equation}
            \max_{\xhat\in \mathbf{\xhat}}|F(\xhat) - p(\xhat)|_d \leq \epsilon.
        \end{equation}

        Thus, the absolute distance $|F(\xhat) - p(\xhat)|$, $\forall \xhat \in [\xhat_i, \xhat_{i+1}]$ satisfies:
        \begin{equation}\label{eq:begin-bound}
            \begin{aligned}
                & |F(\xhat) - p(\xhat)| \\
                & = |F(\xhat) + F(\xhat_i) - F(\xhat_i) + p(\xhat_i) - p(\xhat_i) - p(\xhat)| \\ 
                & \leq |F(\xhat) - F(\xhat_i)| + |F(\xhat_i) - p(\xhat_i)| + |p(\xhat_i) - p(\xhat)| \\
                & \leq C_F \cdot r + C_{p} \cdot r + |F(\xhat_i) - p(\xhat_i)|
            \end{aligned}
        \end{equation}
        
        1) When $|F(\xhat_i)| > \hat{0}$, we have $|F(\xhat_i) - p(\xhat_i)| \leq \epsilon\cdot |F(\xhat_i)|$ (Eq~\ref{eq:rela-dis}). 
        As $|F(\xhat) - F(\xhat_i)| \leq C_F \cdot r$ holds (Eq~\ref{eq:L-bound}), then $|F(\xhat_i)| \leq C_F \cdot r + |F(\xhat)|$, thus from Eq~\ref{eq:begin-bound}, we have:

        \begin{equation*}
            \begin{aligned}
                |F(\xhat) - p(\xhat)| & \leq C_F \cdot r + C_{p} \cdot r + \epsilon\cdot |F(\xhat_i)| \\
                & \leq C_F \cdot r + C_{p} \cdot r + \epsilon\cdot(C_F \cdot r + |F(\xhat)|)\\
                & \leq (C_F + C_{p} + \epsilon\cdot C_F) \cdot r +  \epsilon\cdot |F(\xhat)|, \\
            \end{aligned}
        \end{equation*}
        Let $C = C_F + C_{p} + \epsilon\cdot C_F$, 
        when $|F(\xhat)| > \hat{0}$, we have:
        \begin{equation}\label{eq:res1}
            \begin{aligned}
                |F(\xhat) - p(\xhat)|_d & = \frac{|F(\xhat) - p(\xhat)|}{F(\xhat)} \leq  \frac{C\cdot r}{|F(\xhat)|} + \epsilon, \\
            \end{aligned}
        \end{equation}
        otherwise:
        \begin{equation}\label{eq:res2}
            \begin{aligned}
                |F(\xhat) - p(\xhat)|_d & = |F(\xhat) - p(\xhat)| \leq  C\cdot r.\\
            \end{aligned}
        \end{equation}

        2) When $|F(\xhat_i)| \leq \hat{0}$, we can bound the range of $|F(\xhat)|$ as:
        \begin{equation*}
            \begin{aligned}
                |F(\xhat)| & \leq C_F\cdot r + |F(\xhat_i)| \\
                & \leq C_F\cdot r + \hat{0} < C_F\cdot r + \epsilon, \\
            \end{aligned}
        \end{equation*}
        and we have $|F(\xhat_i) - p(\xhat_i)| \leq \epsilon$, combining with Eq~\ref{eq:begin-bound}:
        \begin{equation*}
            |F(\xhat) - p(\xhat)| \leq C_F \cdot r + C_{p} \cdot r + \epsilon. 
        \end{equation*}
        
        Let $C = C_F +C_{\mathcal{P}}$, when $|F(\xhat)| > \hat{0}$, we have:
        \begin{equation}\label{eq:res3}
            \begin{aligned}
                |F(\xhat) - p(\xhat)|_d & = \frac{|F(\xhat) - p(\xhat)|}{|F(\xhat)|} \\
                & \leq \frac{(C_F + C_{\mathcal{P}})\cdot r}{|F(\xhat)|} + \frac{\epsilon}{|F(\xhat)|} \\
                & \leq \frac{(C_F + C_{\mathcal{P}})\cdot r}{|F(\xhat)|} + \frac{\epsilon}{C_F\cdot r + \epsilon} \\
                & \leq \frac{C\cdot r}{|F(\xhat)|} + 1, \\
            \end{aligned}
        \end{equation}
        otherwise:
        \begin{equation}\label{eq:res4}
            \begin{aligned}
                |F(\xhat) - p(\xhat)|_d = |F(\xhat) - p(\xhat)| \leq C\cdot r + \epsilon.
            \end{aligned}
        \end{equation}

        Combining Eq~\ref{eq:res1},\ref{eq:res2},\ref{eq:res3},\ref{eq:res4}, we prove the result in Theorem~\ref{theorem:precision-bound}.
    \end{proof}




\subsection{Effectiveness of Simulated $\xhat$ through $\mathsf{FLPsimFXP}$ (Algorithm~\ref{flpsimfxp})}\label{proof:FLPsimFXP}
    For all possible FLP encoded input $x \in \mathbb{R}$, the return value $x' = \mathsf{FLPsimFXP}(x, n, f)$ can be represented in $\langle n,f \rangle$-FXP.

    $\forall x \in \mathbb{R}$, if it returns in the first two steps of $\mathsf{FLPsimFXP}$  (Algorithm~\ref{flpsimfxp}), then it is obvious. Otherwise, it does not return in the first two steps of $\mathsf{FLPsimFXP}$, it satisfies that:
    \begin{equation}\label{eq:x-range}
        2^{-f} \leq |x| \leq {2^{n-f-1}}
    \end{equation}

    As $\mathsf{round}_2(x, f)$ rounds-off all the bits beyond $f$ bit after decimal point to $0$, we have that: $\bar{x} = \mathsf{round}_2(x, f) \cdot 2^{f}$ is equivalent to an integer as all bits after the decimal point is $0$. 
    
    As $|\bar{x}| \leq |x| \cdot 2^{f}$, the range of $|\bar{x}|$ satisfies $0 \leq |\bar{x}| \leq 2^{n-1}$. Thus, there exist an $n$-bit integer $\bar{x}$ representing $x'$ where $x' = \bar{x} \cdot 2^{-f}$, which is representable in $\langle n,f \rangle$-FXP number format.

\subsection{Complexity Analysis of OPPE (Algorithm\ref{Algo:OPPE})}\label{sec:oppe-complexity}

    \begin{theorem}[Complexity of OPPE Algorithm]
        The \mulop's complexity is $O(km + k\log{k})$ and \gtop's complexity is $O(m)$.
    \end{theorem}
    \begin{proof}
        In OPPE Algorithm\ref{Algo:OPPE}, there are totally $m$ \gtop's, in Line1. Thus the complexity of \gtop's is quite straightforward, which is $O(m)$.
        
        For complexity of \mulop's, which comes from three parts: 1) Line 3-6, $(2km)$ plaintext-ciphertext \mulop's; 2) Line 8, $(2k)$ ciphertext \mulop's and 3) Line 7, from \code{CalculateKx} function.
        In \code{CalculateKx} function, there are totally $(\lfloor\log{k}\rfloor+1)$ rounds of \mulop's. In each $i$ round, there are $((k+1)-2^{i-1})$ ciphertext \mulop's, $i = 1, ..., \lfloor\log{k}\rfloor+1$. Suppose there are totally $n$ \mulop's in \code{CalculateKx}, we have
        \begin{equation}
            \begin{aligned}
                n & = \sum_{i=1}^{i=\lfloor\log{k}\rfloor+1}{(k+1 - 2^{i-1})} \\
                & = (\lfloor\log{k}\rfloor+1)(k+1) - 2^{\lfloor\log{k}\rfloor + 1} + 1 \\
                & \approx (k+1)\log{k} - k + 2
            \end{aligned}
        \end{equation}
        As the cost of plaintext-ciphertext \mulop is equal or less than ciphertext-ciphertext \mulop in most MPC platforms, the complexity of \mulop's is $O((k+1)\log{k} - k + 2 + 2k + 2km) = O(klog{k} + km)$.
    \end{proof}


\subsection{Obliviousness of OPPE (Algorithm~\ref{Algo:OPPE})}\label{appendix:oblivious-oppe}

    An \emph{oblivious} algorithm means that the execution path is independent of the inputs. 
    In multi-party computation, such property refers to a deterministic and data-independent sequence of executing secure operations, like \emph{secure oblivious sort algorithm}~\cite{hamada2014oblivious} and \emph{secure oblivious data access algorithm}~\cite{keller2014efficient}. 
    The obliviousness property of the OPPE Algorithm~\ref{Algo:OPPE} is quite straightforward as there are no branches based on the inputs or any variables calculated directly or indirectly from the inputs. Formally, we have

    \begin{theorem}[Obliviousness of OPPE Algorithm~\ref{Algo:OPPE}]
        With subroutines $\mathsf{ADD}, \mathsf{MUL}$ and $\mathsf{GT}$ work as black boxes evaluating secure addition, multiplication and greater-than, 
        the execution path of OPPE Algorithm~\ref{Algo:OPPE} is independent of the inputs.
    \end{theorem}

    \par{Proof sketch.  }  For any two representable input values $(\hat{x}, \hat{x'})$, no differences in the evaluation path will be introduced outsides the subroutines ($\mathsf{ADD}, \mathsf{MUL}$ and $\mathsf{GT}$ which work as black boxes).




\section{Security Analysis}~\label{sec:formal-security}

    We adopt the same definitions with Canetti's work~\cite{canetti2000security} in this section. For completeness, we first introduce the security paradigm for readers not familiar with this area, and then give the security definitions and proof of \sysname's protocols.

\subsection{Formal Security Definition}\label{sec:security-defination}

    In the \emph{Real-Ideal} proof paradigm introduced in~\cite{canetti2000security}, two processes \emph{ideal} and \emph{real} are defined. 
    
    In the Ideal-process, we assume a trusted third party exists, who receives the inputs from all parties, and evaluates the target function $f$ locally and distributes the designated results to each party; the ideal $t$-\emph{limited} adversary $\mathcal{A}^{*}$ controls a set of at most $t$ corrupted parties, learns their identities, inputs, internal states, received outputs and can modify their inputs to arbitrary value based on the gathered information. 
    
    In the real-process, parties interact with each other according to a protocol $\pi$ in the presence of a real $t$-limited adversary $\mathcal{A}$, who controls a set of at most $t$ corrupted parties in some adversarial model (\eg semi-honest / malicious, adaptive / non-adaptive). 
    At the end of the computation, the real adversary can let the corrupted parties output some arbitrary value. 
    
    There are several adversarial models to use.  For example, \emph{adaptive} vs. \emph{non-adaptive}.  
    To model an adaptive adversary, people introduce an \emph{environment identity} $\mathcal{Z}$ that can see the inputs, internal states and outputs of all the parties and interacts with the adversary during the evaluation for both ideal and real-processes.  On the other hand, to model a non-adaptive adversary, the adversary cannot interact with $\mathcal{Z}$, nor can it change the member of corrupted parties during the execution.
    
    We define our security in the adaptive model (including both semi-honest and malicious), which is a stronger security definition.
    Let {IDEAL}$_{f, \mathcal{A}^{*}, \mathcal{Z}}$ denote the distribution ensemble of all the parties outputs under any valid security parameter, inputs and randomness in the ideal-process; and let EXEC$_{\pi, \mathcal{A}, \mathcal{Z}}$ denote the same distribution ensemble in the real-process.  In the real-ideal paradigm, $\pi$ securely evaluate function $f$ if it \emph{emulates} the ideal-process in the real-process under any effect of the real adversary $\mathcal{A}$ that can be achieved by \emph{some} ideal adversary $\mathcal{A}^{*}$. We formally define the \emph{secure evaluation} with the following Definition~\ref{def:secure-protocol}:

    \begin{definition}[Secure Evaluation]\label{def:secure-protocol}
        Let $f$ be an $n$-party function and let $\pi$ be a protocol for $n$ parties. We say that protocol $\pi$ adaptively \textbf{$t$-securely evaluates} $f$ if for any adaptive $t$-limited real adversary $\mathcal{A}$ and any environment $\mathcal{Z}$, there exists an adaptive ideal-process adversary $\mathcal{A}^{*}$ whose running time is polynomial in the running time of $\mathcal{A}$, such that

        \begin{equation}
            {IDEAL}_{f, \mathcal{A}^{*}, \mathcal{Z}} \overset{c / s}{\approx} {EXEC}_{\pi, \mathcal{A}, \mathcal{Z}},
        \end{equation}
    	where $\overset{c / s}{\approx}$ means two distribution ensembles computationally or statistically indistinguishable (suitable for computational-limited or unlimited adversary). 
    \end{definition}


\subsection{Security Property Proof of NFGen}\label{proof:security}

    In this section, we first briefly introduce the Canetti's \emph{composition theorem}~\cite{canetti2000security} and use it to prove the security preserving property of \sysname. Note that the following definitions suit for both semi-honest and malicious adversaries~\cite{canetti2000security}.

    \para{Secure composition.  } One commonly used method in developing complex high-level secure protocols for some task is the \emph{modular composition}~\cite{canetti2000security}. We firstly design the high-level protocol by assuming that a series of simple sub-protocols can be carried out securely. Then we design each secure sub-protocol meeting the security guarantee and plug them as subroutines in the high-level protocol. 
    The \emph{composition theorem} states that, if the high-level protocol can securely evaluate (as defined in Definition~\ref{def:secure-protocol}) its function with ideal sub-protocols, then the security and functionality maintained by replacing all the ideal sub-protocols into subroutines.

    \begin{theorem}[secure composition theorem (Collary12 in~\cite{canetti2000security})]\label{theorem:composition}
        Let $t < n$, let $m \in N$ and let $f_1, ..., f_m$ be $n$-party functions. Let $\pi$ be an n-party protocol that adaptively t-securely evaluates $g$ in the $(f_1, ..., f_m)-$hybrid model and assumes that no more than one ideal evaluation call is made at each round. Let $\pi_1, ..., \pi_m$ be n-party protocols that adaptively t-securely evaluate $f_1, ..., f_m$. Then the protocol $\pi^{\pi_1, ..., \pi_m}$ adaptively t-securely evaluates $g$.
    \end{theorem}

    The $(f_1, ..., f_m)-$hybrid model means that the joint parties have the access to call ideal functions $f_1, ..., f_m$.

    \para{Security-preserving property of \sysname. } The generated protocol of \sysname has the property that it guarantees the same \emph{security property} with three provided subroutines $\mathsf{ADD}, \mathsf{MUL}$ and $\mathsf{GT}$ evaluating secure addition, multiplication and greater-than. 
    Formally, it has 

    \begin{theorem}[Security preserving property of NFGen]
        Let $t < n$, and let $f_+, f_{\times}, f_{>}$ be $n$-party functions evaluating addition, multiplication and greater-than. Let $\pi_{+}, \pi_{\times}, \pi_{>}$ be n-party protocols that t-securely evaluate $f_+, f_{\times}, f_{>}$. Then the protocol $\pi_{OPPE}^{\pi_{+}, \pi_{\times}, \pi_{>}}$ generated by \sysname t-securely evaluates \pkm.
    \end{theorem}
    
    \par{Proof sketch.  } 
    Firstly, the protocol $\pi_{OPPE}$ generated through OPPE Algorithm~\ref{Algo:OPPE} can securely evaluate \pkm in the $(f_+, f_{\times}, f_>)-$hybrid model. 
    The protocol $\pi_{OPPE}$ can be constructed by replacing subroutines $\mathsf{ADD}, \mathsf{MUL}$ and $\mathsf{GT}$ with $f_+, f_{\times}$ and $f_>$ following OPPE Algorithm~\ref{Algo:OPPE}.  Since OPPE Algorithm~\ref{Algo:OPPE} organizes each subroutine $\mathsf{ADD}, \mathsf{MUL}$ and $\mathsf{GT}$ sequentially without revealing any information nor introducing any interactions among parties, any cheating behaviors will only happen inside the subroutines. As these subroutines in $\pi_{OPPE}$ are ideal functions $f_+, f_{\times}$ and $f_>$, thus the real-process distribution ensemble ${EXEC}^{f_+, f_{\times}, f_>}_{\pi_{OPPE}, \mathcal{A}, \mathcal{Z}}$ in the hybrid-model is indistinguishable to the ideal-process distribution ensemble ${IDEAL}_{\hat{p}_{k}^{m}, \mathcal{A}^{*}, \mathcal{Z}}$. 
    Then, by the composition theorem~\ref{theorem:composition}, with protocols $\pi_{+}, \pi_{\times}, \pi_{>}$ that securely evaluate $f_+, f_{\times}$ and $f_>$, $\pi_{OPPE}$ can securely evaluate \pkm by replacing each ideal function calls to the corresponding protocols.

\section{NFGen Code Examples}\label{sec:code-example-NFCG}
    In this section, we give a detailed code example to demonstrate the workflow of \sysname, here the selected case is privacy-preserving LR requiring \textit{sigmoid}. 

    The NFD config is shown in Code~\ref{code:NFD}, containing the function expression(\code{function}), target domain $[a, b]$(\code{range}), accuracy threshold $\epsilon$ and $\hat{0}$(\code{tol} and \code{zero\_mask}), number representaion $\langle n,f \rangle$. In this case, it also provide supported operations with profiled time (\code{time\_dict}) and generated performance model(\code{profiler}), these two configurations can also be offered in a separated PPD file. Also, the user can select or offer corresponding code templet(\code{code\_templet}) and set the output file path(\code{config\_file}). Then the user can generate specific code by revoke 
    \code{generate\_nonlinear\_config} as Code~\ref{code:generation}. The generated code is shown in Code block~\ref{code:generation}, which can be directly executed in MP-SPDZ environment.

    \begin{mypython}[caption={NFD \& PPD config demo}, label=code:NFD]
import NFGen.CodeTemplet.templet as temp
import NFGen.PerformanceModel.time_ops as to
import sympy as sp

def sigmoid(x):
    # mpc_exp: lambda x:sp.exp(x)
    # mpc_reci: lambda x:1/x, indicating cipher operator.
    return 1 * mpc_reci(1 + mpc_exp(-x))

config_sigmoid = { # NFD Config
    'function': sigmoid,
    'range': (-8, 10),
    'tol': 1e-3,
    'n': 96,
    'f': 48,
    # soft zero.
    'zero_mask': 1e-6,
    # Set the value out of range.
    'default_values': (0, 1), 
    # Supported operations.
    'ops': ['mpc_exp', 'mpc_reci', 'mpc_log', 'mpc_sqrt'],
    # SPDZ code templet.
    # PrivPy templet use: temp.templet_privpy_cpp.
    'code_templet': temp.templet_spdz,
    # Output file.
    'config_file': './config_sigmoid.py',
    # PPD part
    # Time for basic ops(identify mpc_func operators).
    'time_dict': to.basic_time['Rep3'],
    # Profiler model.
    'profiler': '../PerformanceModel/Rep3_kmProfiler.pkl'
}
    \end{mypython}

    \begin{mypython}[caption={Code generation}, label=code:generation]
from NFGen.main import generate_nonlinear_config

# Pass the config and generate code.
generate_nonlinear_config(config_sigmoid)
    \end{mypython}

    \begin{mypython}[caption={Generated Code (MP-SPDZ)}, label=code:generation]
@types.vectorize
def sigmoid(x):
    """Version2 of general non linear function.

    Args:
        x (Sfixed): the input secret value.
    Returns:
        Sfixed: secret f(x).
    """
    
    # In user-level mpc file:
    # probability trunction acceleration.
    program.use_trunc_pr = True
    program.use_split(3)

    # Config of piece-wise polynomial
    breaks = [-1009.0, -10.0, -7.5, -5.0, -2.5, -1.25, 0.0, 
        1.25, 10.0]
    coeffA = [[0.0, 0.0, 0.0, 0.0, 0.0, 0.0], 
        [584022.5019194265, 284883.4701294364, 56450.73475537558, 
        5664.37093087678, 287.16193188632, 5.87319252822], 
        [1715838.8184591327, 1068788.52768671, 274863.81230151834, 
        36252.85955923422, 2439.47176809213, 66.71430835583], 
        [2506924.277852222, 1808567.6009333746, 551644.7862473574, 
        87996.36476823808, 7267.35497880345, 246.29936351578], 
        [2098737.027154335, 1041805.6678011644, -34419.2181203451, 
        -140100.88887670645, -37953.67734930042, -3405.3812230294], 
        [2097139.5017196543, 1054758.5017103767, 29315.2376400353,
        -56099.88246721192, 0.0, 0.0], 
        [2097139.5017196543, 1054758.5017103783, -29315.2376400381, 
        -56099.88246721102, 0.0, 0.0], 
        [1857186.3089880324, 1582070.1166638373, -437353.1379880485,
        60781.5325081042, -4203.393886761, 115.03921274796], 
        [1.0, 0.0, 0.0, 0.0, 0.0, 0.0]]
    scaler = [[1.0, 1.0, 1.0, 1.0, 1.0, 1.0], 
        [2.3842e-07, 2.3842e-07, 2.3842e-07, 2.3842e-07, 2.3842e-07, 
        2.3842e-07], [2.3842e-07, 2.3842e-07, 2.3842e-07, 
        2.3842e-07, 2.3842e-07, 2.3842e-07], [2.3842e-07, 
        2.3842e-07, 2.3842e-07, 2.3842e-07, 2.3842e-07, 
        2.3842e-07], [2.3842e-07, 2.3842e-07, 2.3842e-07, 
        2.3842e-07, 2.3842e-07, 2.3842e-07], [2.3842e-07, 
        2.3842e-07, 2.3842e-07, 2.3842e-07, 1.0, 1.0], 
        [2.3842e-07, 2.3842e-07, 2.3842e-07, 2.3842e-07, 1.0, 1.0], 
        [2.3842e-07, 2.3842e-07, 2.3842e-07, 2.3842e-07, 
        2.3842e-07, 2.3842e-07], [1.0, 1.0, 1.0, 1.0, 1.0, 1.0]]
    
    m = len(coeffA)
    k = len(coeffA[0])
    degree = k-1

    # Compute the target mask.
    comp = sfix.Array(m)
    for i in range(m):
        comp[i] = (x >= breaks[i])
    cipher_index = bb.get_last_one(comp) 
    
    # Calculate [1, x, x^2, ..., x^{k}].
    pre_muls = floatingpoint.PreOpL(lambda a,b,_: 
        a * b, [x] * degree)

    # Compute c_i*x^i*s_i.
    poss_res = [0]*m
    for i in range(m):
        poss_res[i] = coeffA[i][0] * scaler[i][0]
        for j in range(degree):
            poss_res[i] += (coeffA[i][j+1] 
            * pre_muls[j] * scaler[i][j+1])

    # Get result with mask and all possible values.
    return sfix.dot_product(cipher_index, poss_res)
    \end{mypython}

     \footnotesize
     \begin{lstlisting}[caption={Generated Code (PrivPy C++ Code)}, label=code:generation-pp]
bool SS4Runner::sigmoid(const size_type length,
    TypeSet::FNumberArr *result,
    const TypeSet::FNumberArr *num,
    bool use_current_thread) {
    check_runner_terminate_status(false);
    TypeSet::FNUMT *result_x = result->get_x(), *result_x_ = result->get_x_();

    const size_type K = 9;
    const size_type config_length = length * K;
    TypeSet::FNumberArr coeff(config_length);
    TypeSet::FNumberArr scaler(config_length);
    TypeSet::FNumberArr x_items(config_length);

    sigmoid_calculateCoeff (length, &coeff, &scaler, num); });
    calculateKx(length, &x_items, num, K); });

    mul<TypeSet::FNUMT>(config_length, &x_items, &x_items, &coeff);
    mul<TypeSet::FNUMT>(config_length, &x_items, &x_items, &scaler);

    double *ftmp2 = new double[length];
    for (int i = 0; i < length; i++) ftmp2[i] = 0.0;

    map2numv<double, TypeSet::FNUMT>(ftmp2, length, result);
    for (size_type i = 0; i < length; i++) {
        for (size_type j = 0; j < K; j++) {
            result_x[i] += x_items.x[i * K + j];
            result_x_[i] += x_items.x_[i * K + j];
        }
    }
    delete[] ftmp2;
    return true;
}

bool SS4Runner::sigmoid_calculateCoeff(const size_type length,
    TypeSet::FNumberArr *coeff,
    TypeSet::FNumberArr *scaler,
    const TypeSet::FNumberArr *num) {
    check_runner_terminate_status(false);
    TypeSet::FNUMT *num_x = num->get_x(), *num_x_ = num->get_x_();

    const size_type M = 7;
    const size_type K = 9;
    const double Breaks[M] = {-1009.0, -10.0, -5.0, -1.25, 0.0, 1.25, 10.0};
    const double CoeffA[M * K] = {0.0, 0.0, 0.0, 0.0, 0.0, 0.0, 0.0, 0.0, 0.0, 2736794.311829414, 2166127.0649512047, 779033.1999940014, 165208.365427917, 22462.75021140543, 1994.81274157933, 112.49785815234, 3.6702035068, 0.05287661137, 2060712.4699856702, 909977.1725834787, -224423.69609669817, -287140.3378066692, -103632.73876614487, -20259.29529210076, -2311.60919005347, -144.85798280994, -3.84685342155, 2097139.5017196543, 1054758.5017103767, 29315.2376400353, -56099.88246721192, 0.0, 0.0, 0.0, 0.0, 0.0, 2097139.5017196543, 1054758.5017103783, -29315.2376400381, -56099.88246721102, 0.0, 0.0, 0.0, 0.0, 0.0, 2020005.5083353834, 1217129.597655899, -113822.68074184433, -88781.97063835277, 35521.12204303316, -6139.86216019294, 574.93708725425, -28.4375893287, 0.58380827441, 1.0, 0.0, 0.0, 0.0, 0.0, 0.0, 0.0, 0.0, 0.0};
    const double Scaler[M * K] = {1.0, 1.0, 1.0, 1.0, 1.0, 1.0, 1.0, 1.0, 1.0, 2.3842e-07, 2.3842e-07, 2.3842e-07, 2.3842e-07, 2.3842e-07, 2.3842e-07, 2.3842e-07, 2.3842e-07, 2.3842e-07, 2.3842e-07, 2.3842e-07, 2.3842e-07, 2.3842e-07, 2.3842e-07, 2.3842e-07, 2.3842e-07, 2.3842e-07, 2.3842e-07, 2.3842e-07, 2.3842e-07, 2.3842e-07, 2.3842e-07, 1.0, 1.0, 1.0, 1.0, 1.0, 2.3842e-07, 2.3842e-07, 2.3842e-07, 2.3842e-07, 1.0, 1.0, 1.0, 1.0, 1.0, 2.3842e-07, 2.3842e-07, 2.3842e-07, 2.3842e-07, 2.3842e-07, 2.3842e-07, 2.3842e-07, 2.3842e-07, 2.3842e-07, 1.0, 1.0, 1.0, 1.0, 1.0, 1.0, 1.0, 1.0, 1.0};

    // 1. Calculate the corresponding index
    // a) Expand num and breaks to outter_comparision.
    const size_type expand_length = length * M;
    TypeSet::FNumberArr findex_arr(expand_length);
    outter_gt<double, TypeSet::FNUMT>(Breaks, num, expand_length, &findex_arr);
    // b) Compute the target mask
    for (size_type i = 0; i < length; i++) {
        memcpy(ctmp + i * M, findex_arr + i*M+1, sizeof(TypeSet::FNUMT) * (M - 1));
    }
    sub<TypeSet::FNUMT>(expand_length, &findex_arr, &findex_arr, &ctmp);
    // 2. Fetch out target coeff and scaler.
    const size_type config_length = length * K;
    TypeSet::FNumberArr coeff(config_length);
    TypeSet::FNumberArr scaler(config_length);
    size_type shape1[2] = {length, M};
    size_type shape2[2] = {M, K};
    inner_product<TypeSet::FNUMT, double>(
        &coeff, shape1, shape2, &findex_arr, CoeffA);
    inner_product<TypeSet::FNUMT, double>(
        &scaler, shape1, shape2, &findex_arr, Scaler);

    return true;
}
     \end{lstlisting}

     \footnotesize
     \begin{mypython}[caption={Generated Code (PrivPy Python Code)}]
@pp.local_import("numpy", "np")
def sigmoid(x):
    import pnumpy as pnp
    
    def _calculate_kx(x, k):
        items = pnp.transpose(pnp.tile(x, (k, 1)))
        items[:, 0] = pp.sfixed(1)

        shift = 1
        while shift < k:
            items[:, shift:] *= items[:, :len(items[0])-shift]
            shift *= 2
        return items

    def _fetch_index(x, breaks):
        if isinstance(x, pp.FixedArr):
            x = pnp.transpose(pnp.tile(x, (len(breaks), 1)))
        breaks = np.tile(breaks, (len(x), 1))

        cipher_comp = x >= breaks
        cipher_index = pnp.util.get_last_one(cipher_comp)

        return cipher_index
    
    # same breaks, coeffA and scaler with other generated code.

    breaks = np.array(breaks)
    coeffA = np.array(coeffA)
    scalerA = np.array(scaler)

    k = int(len(coeffA[0]))
    cipher_index = _fetch_index(x, breaks)
    coeff = pnp.dot(cipher_index, coeffA)
    scaler = pnp.dot(cipher_index, scalerA)
    x_items = _calculate_kx(x, k)
    tmp_res = x_items * coeff
    res = pnp.sum(tmp_res * scaler, axis=1)
    return res
     \end{mypython}

\section{Full Micro-benchmark Results}\label{sec:full-microbenchmark}

\begin{table*}
	\caption{Micro-benchmark for probability distribution functions}\label{Tab:full2}
	\footnotesize
	\resizebox{\textwidth}{!}{
		\begin{threeparttable}
			\centering
			\begin{tabular}{c|c|c|c|c|c|c|c|c|c|c}
				\midrule[1.1pt]
				\multirow{2}{*}{$F(x)$} & \multirow{2}{*}{$\mathcal{S}$} & \multirow{2}{*}{\checkmark}  & \multirow{2}{*}{$(k, m)$} & \multirow{2}{*}{$T_{\text{Fit}}$} & \multicolumn{3}{c|}{Communication(MB)} & \multicolumn{3}{c}{Efficiency(Ms)} \\ \cline{6-11}
                &  &  &  &  & Base & \scriptsize{\sf{NFGen}} & \textbf{Save} & Base & \scriptsize{\sf{NFGen}} & \textbf{SpeedUp} \\[1pt] \hline\hline
				\multirow{6}*{\shortstack{$\textit{Normal\_dis}(x) = \frac{e^{-\frac{x^2}{2}}}{\sqrt{2\pi}}$ \\  $x \in [-10, +10]$, $F(x) \in [0.0, 0.4]$ \\ Non-linear buildling-blocks: $1$}}  
                & A    &  $\times$  &   (8, 12)     &   5.2  &      420     &       295 &        \textbf{30\%} &             67 &         24 &       \textbf{2.8$\times$}  \\
				& B    &  $\times$  &   (8, 12)     &   3.6  &      3       &         2 &        \textbf{45\%} &           4906 &       156 &       \textbf{31.5$\times$} \\
				& C    &  $\times$  &   (8, 12)     &   3.6  &      7       &         5 &        \textbf{27\%} &           5029 &       970 &        \textbf{5.2$\times$}  \\
				& D    &  $\times$  &   (8, 12)     &   3.6 &       24      &        23 &         \textbf{5\%} &           6588 &      1846 &        \textbf{3.6$\times$}  \\
                & E    &  $\times$  &   (5, 22)     &   3.6  &      257     &       481 &       -87\% &          89740 &    166328 &        0.5$\times$  \\
				& F    &  $\times$  &   (8, 12)     &   3.6  &      249     &       301 &       -21\% &          14908 &     14861 &        1.0$\times$  \\
                \midrule
				\multirow{6}*{\shortstack{$\textit{Cauchy\_dis}(x) = \frac{1} {\pi(1 + x^{2})}$ \\  $x \in [-40, +40]$, $F(x) \in [0.0, 0.3]$ \\ Non-linear buildling-blocks: $1$}}  
				& A    & \checkmark   &   (10, 10)&  4.4      &   202 &       295 &       -50\% &             82 &        26 &        \textbf{3.2$\times$} \\
                & B    & \checkmark    &   NA      &  3.6      &   1   &         1 &         0\% &             69 &        71 &        1.0$\times$  \\
				& C    &  \checkmark   &   NA      &  3.6      &   2   &         2 &         0\% &            405 &       403 &        1.0$\times$  \\
				& D    & \checkmark    &   NA      &  3.6      &   8   &         8 &         0\% &            698 &       689 &        1.0$\times$  \\
                & E    & \checkmark    &   NA      &  3.6      &   50  &        50 &         0\% &          15496 &     15677 &        1.0$\times$  \\
				& F    & \checkmark    &   NA      &  3.6      &   47  &        47 &         0\% &           2248 &      2243 &        1.0$\times$  \\
				\midrule
				\multirow{6}*{\shortstack{$\textit{Gamma\_dis}(x) = \frac{x^{\gamma-1}e^{-x}}{\Gamma(\gamma)}$ \\ $\gamma = 0.5$, $x \in [0.0, 50]$, $F(x) \in [0.0, 0.4]$ \\ Non-linear buildling-blocks: $2$}}  
                & A    &  $\times$  &   (7, 19)     &   5.9     &   793 &       393 &       \textbf{50\%}  &           137  &        30 &        \textbf{4.6$\times$}  \\
				& B    &  $\times$  &   (7, 21)     &   4.3     &   3 &         2   &        \textbf{26\%} &           4624 &       216 &       \textbf{21.4$\times$}  \\
				& C    &  $\times$  &   (5, 27)     &   4.3     &   6 &         8   &       -30\% &           5008 &      1443 &        \textbf{3.5$\times$}  \\
				& D    &  $\times$  &   (5, 27)     &   4.3     &   23 &        33  &       -45\% &           7206 &      2695 &        \textbf{2.7$\times$}  \\
                & E    &  $\times$  &   (5, 27)     &   4.3     &   255 &       527 &      -106\% &          89018 &    179739 &        0.5$\times$  \\
				& F    &  $\times$  &   (5, 27)     &   4.3     &   247 &       402 &       -63\% &          14308 &     19226 &        0.7$\times$  \\
				\midrule
				\multirow{6}*{\shortstack{$\textit{Chi\_square\_dis}(x) = \frac{e^{\frac{-x} {2}}x^{\frac{\nu} {2} - 1}}{2^{\frac{\nu} {2}}\Gamma(\frac{\nu} {2}) }$ \\ $v = 4$, $x \in [0.0, 50]$, $F(x) \in [0.0, 0.2]$ \\ Non-linear buildling-blocks: $2$}}  
                & A    & \checkmark  &   (8, 5)  &  2.4  &  419  &       168 &        \textbf{60\%} &             62 &        19 &         \textbf{3.3$\times$}  \\
				& B    & \checkmark    &   (8, 5)  &  1.5  &  3    &         1 &        \textbf{75\%} &           4846 &        72 &       \textbf{67.4$\times$}  \\
				& C    & \checkmark    &   (7, 6)  &  1.5  &  7    &         3 &        \textbf{57\%} &           5074 &       538 &        \textbf{9.4$\times$}  \\
				& D    & \checkmark    &   (7, 6)  &  1.5  & 24    &        13 &        \textbf{46\%} &           6496 &      1016 &        \textbf{6.4$\times$}  \\
                & E    & \checkmark    &   (5, 10) &  1.5  & 257   &       224 &        \textbf{13\%} &          89594 &     77908 &        \textbf{1.2$\times$}  \\
				& F    & \checkmark    &   (7, 6)  &  1.5  &  249  &       140 &        \textbf{44\%} &          14452 &      6904 &        \textbf{2.4$\times$}  \\
				\midrule
				\multirow{6}*{\shortstack{$\textit{Exp\_dis}(x) = e^{-x}$ \\  $x \in [10^{-5}, 10]$, $F(x) \in [0.0, 1.0]$ \\ Non-linear buildling-blocks: $2$}}  
                & A    & \checkmark   &   (9, 3) &  2.9  &  420 &     184 &       \textbf{60\%} &                67 &        21 &        \textbf{3.1$\times$}  \\
				& B    & \checkmark    &   (6, 5) &  1.5  &  3 &         1 &        \textbf{79\%} &            225 &        63 &        \textbf{3.6$\times$}  \\
				& C    & \checkmark    &   (6, 5) &  1.5  &  6 &         2 &        \textbf{61\%} &           1495 &       418 &        \textbf{3.6$\times$}  \\
				& D    & \checkmark    &   (6, 5) &  1.5  &  22 &        10 &        \textbf{53\%} &           2449 &       821 &        \textbf{3.0$\times$}  \\
                & E    & \checkmark    &   (6, 5) &  1.5  & 238 &       132 &        \textbf{45\%} &          83416 &     46803 &        \textbf{1.8$\times$}  \\
				& F    & \checkmark    &   (6, 5) &  1.5  & 231 &       100 &        \textbf{57\%} &          11527 &      4872 &        \textbf{2.1$\times$}  \\
				\midrule
				\multirow{6}*{\shortstack{$\textit{Log\_dis}(x) = \frac{e^{-((\ln x)^{2}/2\sigma^{2})}}{x\sigma\sqrt{2\pi}}$ \\ $\sigma=1.0$, $x \in [10^{-4}, 20]$, $F(x) \in [0.0, 0.7]$ \\ Non-linear buildling-blocks: $3$}}  
				& A     & $\times$  &  (10, 10)    &   4.3     &  2039     &       295 &        \textbf{90\%} &            456 &        25 &        \textbf{18.5$\times$}  \\
                & B     & \checkmark    &  (8, 12)     &   3.5     &  8        &         2 &        \textbf{80\%} &            486 &       149 &        \textbf{3.3$\times$}  \\
				& C     & \checkmark    &  (8, 12)     &   3.5     &  12       &         6 &        \textbf{49\%} &           3088 &      1043 &        \textbf{3.0$\times$}  \\
				& D     & \checkmark    &  (8, 12)     &   3.5     & 37        &        26 &        \textbf{29\%} &           4665 &      2026 &        \textbf{2.3$\times$}  \\
                & E     & \checkmark    &  (6, 17)     &   3.5     & 449       &       433 &         \textbf{4\%} &         142087 &    151929 &        0.9$\times$  \\
				& F     & \checkmark    &  (8, 12)     &   3.5     &  431      &       330 &        \textbf{23\%} &          22545 &     16212 &        \textbf{1.4$\times$}  \\

				\midrule
				\multirow{6}*{\shortstack{$\textit{Bs\_dis}(x)^{1)}= \left (\frac{\sqrt{x} + \sqrt{\frac{1} {x}}} {2\gamma x} \right) 
						\phi \left (\frac{\sqrt{x} - \sqrt{\frac{1} {x}}} {\gamma} \right)$ \\ $\gamma=0.5$, $x \in [10^{-6}, 30]$, $F(x) \in [0.0, 0.2]$ \\ Non-linear buildling-blocks: $3$}}  
                & A  & $\times$  &     (10, 8) &   4.0     &   2815 &      263     &       \textbf{90\%} &            630 &         22 &        \textbf{29.1$\times$}  \\
                & B  & $\times$   &     (7, 11) &   3.2     &   13 &         1      &        \textbf{89\%} &          11463 &       133 &       \textbf{86.1$\times$}  \\
				& C  & $\times$   &     (5, 16) &   3.2     &   23 &         5      &        \textbf{79\%} &          14631 &       915 &       \textbf{16.0$\times$}  \\
				& D  & $\times$   &     (5, 16) &   3.2     &   65 &        22      &        \textbf{66\%} &          19167 &      1763 &       \textbf{10.9$\times$}  \\
                & E  & $\times$   &     (5, 16) &   3.2     &   741 &       352     &        \textbf{53\%} &         239549 &    122325 &        \textbf{2.0$\times$}  \\
				& F  & $\times$   &     (5, 16) &   3.2     &   718 &       268     &        \textbf{63\%} &          42157 &     13136 &        \textbf{3.2$\times$}  \\
				\midrule    
			\end{tabular}
			\begin{tablenotes}
				\item * $T_{\text{Fit}}$ means the time for \pkm fitting, we generate candidate \Pkm with $k$ ranging from $[3, 10]$. The second column (\checkmark) means whether baseline function achieves the same accuracy threshold ($\epsilon=10^{-3}$ with $\hat{0}=10^{-6}$).
				\item 1) \textit{Bs\_dis} means \textit{Birnbaum-Saunders(Fatigue Life) probability distribution}~\cite{birnbaum1969new}.
			\end{tablenotes}
	\end{threeparttable}}
\end{table*}

\begin{table*}
\caption{Micro-benchmark for activation functions}\label{Tab:full1}
\footnotesize
\resizebox{\textwidth}{!}{
\begin{threeparttable}
    \centering
    \begin{tabular}{c|c|c|c|c|c|c|c|c|c|c}
    \midrule[1.1pt]
    \multirow{2}{*}{$F(x)$} & \multirow{2}{*}{$\mathcal{S}$} & \multirow{2}{*}{\textbf{\checkmark}}  & \multirow{2}{*}{$(k, m)$} & \multirow{2}{*}{$T_{\text{Fit}}$} & \multicolumn{3}{c|}{Communication(MB)} & \multicolumn{3}{c}{Efficiency(Ms)} \\ \cline{6-11}
    &  &  &  &  & Base & \scriptsize{\sf{NFGen}} & \textbf{Save} & Base & \scriptsize{\sf{NFGen}} & \textbf{SpeedUp} \\[1pt] \hline\hline
    \multirow{6}*{\shortstack{$sigmoid(x) = \frac{1}{1 + e^{-x}}$ \\  $x \in [-50, +50]$, $F(x) \in [0.0, 1.0]$ \\ Non-linear buildling-blocks: $2$}}  
    & A    & $\times$ &     (10, 8)     &  4.3    &     618 &       263 &        \textbf{60\%} &            147 &        23 &      \textbf{6.3$\times$} \\
    & B    & \checkmark &     (7, 10)     &  3.5    &     1   &         1 &        -5\% &            137 &       124 &      \textbf{1.1$\times$} \\
    & C    & \checkmark &     (5, 14)     &  3.5    &     4   &         4 &        -5\% &           1155 &       802 &      \textbf{1.4$\times$} \\
    & D    & \checkmark &     (5, 14)     &  3.5    &     18  &        19 &        -8\% &           1863 &      1525 &      \textbf{1.2$\times$} \\
    & E    & \checkmark &     (5, 14)     &  3.5    &     212 &       308 &       -45\% &          75949 &    106857 &      0.7$\times$ \\
    & F    & \checkmark &     (5, 14)     &  3.5    &     207 &       234 &       -13\% &           9732 &     11224 &      0.9$\times$ \\
    \midrule
    \multirow{6}*{\shortstack{$tanh(x) = \frac{e^{x} - e^{-x}}{e^{x} + e^{-x}}$ \\  $x \in [-50, +50]$, $F(x) \in [-1.0, 1.0]$ \\ Non-linear buildling-blocks: $3$}}  
    & A    & $\times$ &     (9, 8)  &   4.5     &   1876    &       216 &       \textbf{90\%} &            335 &        21 &       \textbf{15.7$\times$} \\
    & B    & $\times$ &     (5, 9)  &   3.2     &   13      &         1 &       \textbf{92\%} &            800 &        80 &       \textbf{10.0$\times$} \\
    & C    & $\times$ &     (5, 9)  &   3.2     &   19      &         3 &       \textbf{83\%} &           5901 &       597 &       \textbf{9.9$\times$} \\
    & D    & $\times$ &     (5, 9)  &   3.2     &   64      &        14 &       \textbf{78\%} &           8882 &      1115 &        \textbf{8.0$\times$} \\
    & E    & $\times$ &     (5, 9)  &   3.2     &   996     &       197 &       \textbf{80\%} &         337530 &     68550 &        \textbf{4.9$\times$} \\
    & F    & $\times$ &     (5, 9)  &   3.2     &   966     &       150 &       \textbf{84\%} &          45486 &      7309 &        \textbf{6.2$\times$} \\
    \midrule
    \multirow{6}*{\shortstack{$\textit{soft\_plus}(x) = \log(1 + e^{x})$ \\  $x \in [-20, 50]$, $F(x) \in [0.0, 49.9]$ \\Non-linear buildling-blocks: $2$}}  
    & A    & $\times$ &     (10, 7) &   3.6  &      1127&       248 &        \textbf{80\%}  &            221 &        23 &        \textbf{9.5$\times$} \\
    & B    & $\times$ &     (8, 9)  &   2.6  &      5   &         1 &        \textbf{78\%}  &            384 &       110 &        \textbf{3.5$\times$} \\
    & C    & $\times$ &     (6, 11) &   2.6  &      8   &         4 &        \textbf{49\%}  &           2847 &       797 &        \textbf{3.6$\times$} \\
    & D    & $\times$ &     (6, 11) &   2.6  &      27  &        19 &        \textbf{29\%}  &           4054 &      1475 &        \textbf{2.7$\times$} \\
    & E    & $\times$ &     (4, 19) &   2.6  &      318 &       343 &        -8\%  &         105809 &    116529 &        0.9$\times$ \\
    & F    & $\times$ &     (6, 11) &   2.6  &      307 &       220 &        \textbf{28\%}  &          15739 &     10780 &        \textbf{1.5$\times$} \\
    \midrule
    \multirow{6}*{\shortstack{$elu(x) = \begin{cases} x & x > 0 \\ \alpha*(e^{x}-1) & x \leq 0 \end{cases}$ \\  $\alpha = 1.0$, $x \in [-50, 20]$, $F(x) \in [-1.0, 19.9]$ \\ Non-linear buildling-blocks: $2$}}  
    & A    & $\times$ &     (7, 4) &    1.7     &  440  &       153 &        \textbf{70\%}  &             82 &        20 &        \textbf{4.0$\times$} \\
    & B    & $\times$ &     (4, 7) &    0.9     &   3   &         1 &        \textbf{78\%}  &            241 &        66 &        \textbf{3.7$\times$} \\
    & C    & $\times$ &     (4, 7) &    0.9     &  6    &         2 &        \textbf{64\%}  &           1590 &       390 &        \textbf{4.1$\times$} \\
    & D    & $\times$ &     (4, 7) &    0.9     &  22   &         9 &        \textbf{57\%}  &           2540 &       757 &        \textbf{3.4$\times$} \\
    & E    & $\times$ &     (4, 7) &    0.9     &  246  &       127 &        \textbf{48\%}  &          85950 &     43302 &        \textbf{2.0$\times$} \\
    & F    & $\times$ &     (4, 7) &    0.9     &  238  &        96 &        \textbf{60\%}  &          11900 &      4783 &        \textbf{2.5$\times$} \\
    \midrule
    \multirow{6}*{\shortstack{$selu(x) = \lambda \begin{cases} \mbox{$x$} & \mbox{if } x > 0\\ \mbox{$\alpha e^x-\alpha$} & \mbox{if } x \leq 0 \end{cases}$~\cite{klambauer2017self} \\ $\alpha = 1.6732632$ and $\lambda = 1.05007010$ \\$x \in [-50, 20]$, $F(x) \in [-1.8, 20.9]$ \\ Non-linear buildling-blocks: $2$}}  
    & A   & $\times$ &      (7, 4) &    2.5     &   440 &       153 &        \textbf{70\%}  &             85 &        19 &        \textbf{4.5$\times$} \\
    & B   & $\times$ &      (7, 4) &    1.3     &   3   &         1 &        \textbf{82\%}  &            247 &        53 &        \textbf{4.7$\times$} \\
    & C   & $\times$ &      (4, 8) &    1.3     &   6   &         2 &        \textbf{60\%}  &           1664 &       426 &        \textbf{3.9$\times$} \\
    & D   & $\times$ &      (4, 8) &    1.3     &   22  &        10 &        \textbf{54\%}  &           2578 &       860 &        \textbf{3.0$\times$} \\
    & E   & $\times$ &      (4, 8) &    1.3     &   250 &       146 &        \textbf{42\%}  &          86658 &     49364 &        \textbf{1.8$\times$} \\
    & F   & $\times$ &      (4, 8) &    1.3     &   243 &       111 &        \textbf{54\%}  &          11968 &      5402 &        \textbf{2.2$\times$} \\
    \midrule
    \multirow{6}*{\shortstack{$gelu(x) = 0.5x\left(1+\text{tanh}\left(\sqrt{\frac{2}{\pi}}(x+\alpha\,x^3)\right)\right)$ \\ $\alpha = 0.04472$ $x \in [-20, 20]$, $F(x) \in [-0.0, 20.0]$ \\ Non-linear buildling-blocks: $3$}}  
    & A   & $\times$ &      (8, 6) &    1.1  &  267     &       200 &        \textbf{30\%}  &             41 &        21 &        \textbf{2.0$\times$} \\
    & B   & $\times$ &      (4, 9) &    0.6  &   13     &         1 &        \textbf{93\%}  &            800 &        75 &       \textbf{10.7$\times$} \\
    & C   & $\times$ &      (4, 9) &    0.6  &   19     &         3 &        \textbf{87\%}  &           6007 &       522 &       \textbf{11.5$\times$} \\
    & D   & $\times$ &      (4, 9) &    0.6  &  65      &        11 &        \textbf{83\%}  &           9058 &       936 &       \textbf{9.7$\times$} \\
    & E   & $\times$ &      (4, 9) &    0.6  &  1009    &       164 &         \textbf{84\%}  &         344271 &     56269 &        \textbf{6.1$\times$} \\
    & F   & $\times$ &      (4, 9) &    0.6  &  978     &       124 &        \textbf{87\%}  &          46253 &      6109 &        \textbf{7.6$\times$} \\
    \midrule
    \multirow{6}*{\shortstack{$\textit{soft\_sign}(x) = \frac{x}{1 + |x|}$ \\  $x \in [-50, 50]$, $F(x) \in [-1.0, 1.0]$ \\ Non-linear buildling-blocks: $2$}}  
    & A   & $\times$ &      (8, 8)  &  1.9      &       518 &       231 &         \textbf{60\%} &            131 &        21 &        \textbf{6.1$\times$} \\
    & B   & \checkmark &      NA      &  1.3      &       1   &         1 &         0\%  &             79 &        78 &        1.0$\times$ \\
    & C   & \checkmark &      NA      &  1.3      &       2   &         2 &         0\%  &            451 &       437 &        1.0$\times$ \\
    & D   & \checkmark &      NA      &  1.3      &        8  &         8 &         0\%  &            741 &       753 &        1.0$\times$ \\
    & E   & \checkmark &      NA      &  1.3      &       52  &        52 &         0\%  &          15507 &     15520 &        1.0$\times$ \\
    & F   & \checkmark &      NA      & 1.3       &       49  &        49 &         0\%  &           2315 &      2373 &        1.0$\times$ \\
    \midrule
    \multirow{6}*{\shortstack{$isru(x) = \frac{x}{\sqrt{1 + x^2}}$ \\  $x \in [-50, 50]$, $F(x) \in [-1.0, 1.0]$ \\ Non-linear buildling-blocks: $2$}}  
    & A   & $\times$ &      (6, 8)  &  4.4      &  576  &       203 &        \textbf{60\%}  &            157 &        21 &         \textbf{7.4$\times$} \\
    & B   & \checkmark &      (6, 8)  &  3.4      & 3     &         1 &        \textbf{66\%}  &            209 &        96 &        \textbf{2.2$\times$} \\
    & C   & \checkmark &      (4, 13) &  3.4      &   5   &         3 &        0\%  &           1430 &       699 &        \textbf{2.0$\times$} \\
    & D   & \checkmark &      (4, 13) &  3.4      &  15   &        15 &        0\%  &           2088 &      1246 &        \textbf{1.7$\times$} \\
    & E   & \checkmark &      NA      &  3.4      & 145   &       145 &        0\%  &          44751 &     45257 &        1.0$\times$ \\
    & F   & \checkmark &      NA      &  3.4      &  140  &       140 &        0\%  &           7336 &      7399 &        1.0$\times$ \\
    \midrule
    \end{tabular}
    \begin{tablenotes}
        \item * $T_{\text{Fit}}$ means the time for \pkm fitting, we generate candidate \Pkm with $k$ ranging from $[3, 10]$. The second column (\checkmark) means whether baseline function achieves the same accuracy threshold ($\epsilon=10^{-3}$ with $\hat{0}=10^{-6}$).
    \end{tablenotes}
\end{threeparttable}}
\end{table*}


\end{document}